\documentclass[11pt]{article}
\pdfoutput=1
\usepackage{palatcm, color, verbatim}
\usepackage{graphicx,graphics, epsfig, psfrag, subfigure, epstopdf}
\usepackage{algorithm, algorithmic, calc,  url, array}
\usepackage{amssymb, amsmath}
\usepackage{multicol,balance}

\begin{document}
\newtheorem{theorem}{Theorem}
\newtheorem{corollary}{Corollary}
\newtheorem{definition}{Definition}
\newtheorem{lemma}{Lemma}

\newcommand{\define}{\stackrel{\triangle}{=}}

\def\QED{\mbox{\rule[0pt]{1.5ex}{1.5ex}}}
\def\proof{\noindent\hspace{2em}{\it Proof: }}

\date{}
\title{Aligned Interference Neutralization and the\\ Degrees of Freedom of the $2\times 2\times 2$ Interference Channel}
\author{\normalsize Tiangao Gou, Syed A. Jafar, Sang-Woon  Jeon, Sae-Young Chung\\
{\small Email: tgou@uci.edu, syed@uci.edu, swjeon@kaist.ac.kr, sychung@ee.kaist.ac.kr}
       }
%% Notes
\maketitle

\thispagestyle{empty}

%%%%%%%%%%%%%%%%%%%% Abstract %%%%%%%%%%%%%%%%%%%%%%%%%
\begin{abstract}
We show that the $2\times 2\times 2$ interference network, i.e., the multihop interference network formed by concatenation of two 2-user interference channels  achieves the min-cut outer bound value of 2 DoF, for almost all values of channel coefficients, for both time-varying or fixed channel coefficients. The key to this result is a new idea, called aligned interference neutralization, that provides a way to align  interference terms over each hop in a manner that allows them to be cancelled over the air at the last hop.
\end{abstract}

%%%%%%%%%%%%%%%%%%%% Introduction %%%%%%%%%%%%%%%%%%%%%%%%%
\section{Introduction}
Recent years have seen rapid progress in our understanding of the capacity limits of wireless networks. Some of the most remarkable advances have come about in the settings of  (a) multihop multicast, where capacity (within constant gap that is independent of SNR and channel parameters) is  given by the network min-cut \cite{Avestimehr_Diggavi_Tse} and  (b) single hop interference networks, for which a variety of capacity approximations have been obtained in the form of degrees of freedom (DoF) characterizations (e.g., \cite{Jafar_Shamai, Cadambe_Jafar_int, Cadambe_Jafar_X, Motahari_Gharan_Maddah_Khandani}), generalized degrees of freedom (GDOF) (e.g., \cite{Parker_Bliss_Tarokh, Jafar_Vishwanath_GDOF, Bandemer_ElGamal,Huang_Cadambe_Jafar}), O(1) approximations  (e.g. \cite{Gou_Jafar_O1,Gou_Jafar_Wang}), constant gap approximations  (e.g., \cite{Etkin_Tse_Wang, Bresler_Parekh_Tse, Suh_Tse_FB, Tse_Yates}), and exact capacity results (e.g., \cite{MK_int, Shang_Kramer_Chen, Sreekanth_Veeravalli,Sreekanth_Veeravalli_MIMO, Sridharan_Jafarian_Vishwanath_Jafar, Jose_Vishwanath, Cadambe_Jafar_MACZBC, Jafar_ergodic}).
\begin{figure}[!h]
\centering
\includegraphics[width=6.5in]{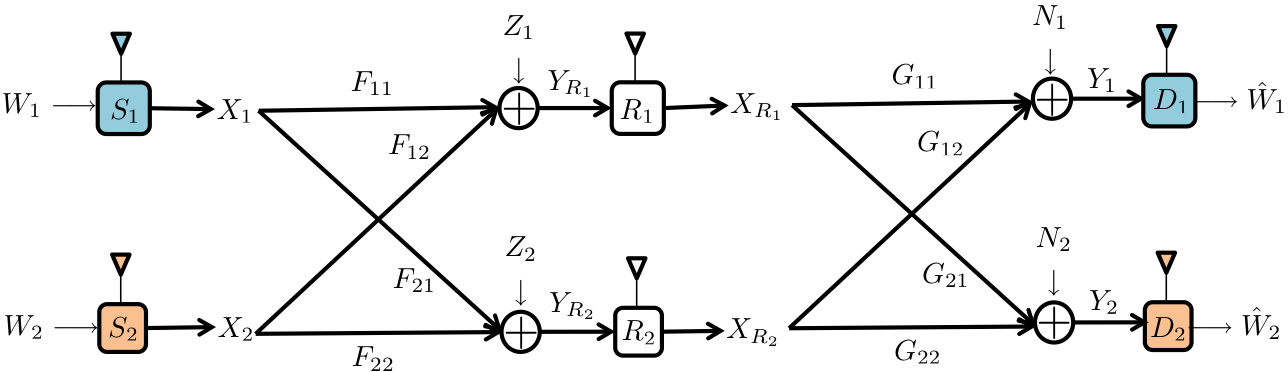}
\caption{$2\times 2\times 2$ IC}
\label{fig:222}
\end{figure}

In spite of the rapid advances in our understanding of multihop multicast and single hop interference networks, relatively little progress has been made so far in our understanding of the fundamental limits of \emph{multihop interference} networks. Of particular interest are layered multihop interference networks formed by concatenation of single hop networks so that each node can only be heard by the nodes in the next layer. Much of the prior work on such multihop interference networks has focused on either the setting where there is a single relay node that is equipped with multiple antennas \cite{Borade_Zheng_Gallager, Tannious_Nosratinia,Chen_Cheng}, or the setting where there is a large number of distributed relay nodes \cite{Boelcskei_Nabar_Oyman_Paulraj,Morgenshtern_Boelcskei, Esli_Wittneben,Rankov_Wittneben, Wittneben_mag}. Relatively little is known about multihop interference networks when the relays are distributed, equipped with only a single antenna each, and there are not many of them. As the simplest example of such a network, the setting shown in Fig. \ref{fig:222} is of fundamental interest. It is remarkable that even a coarse capacity characterization in the form of degrees-of-freedom (DoF) is not available for this network in general. The network of Fig. \ref{fig:222} is the focus of this paper and will henceforth be referred to as the $2\times 2\times 2$ IC. We are interested specifically in the DoF of this network.  We start by reviewing prior work on this channel, which is based on the available insights from the study of the two user interference channel, interference alignment principles and the ideas of distributed zero forcing (interference neutralization).

\subsection*{Interference Channel Approach}
One approach  to the $2\times 2\times 2$ IC is to view it as a cascade of two interference channels. This approach is appealing because the 2 user interference channel is a much-studied problem in network information theory and the abundance of insights developed into this problem, both classical and recent, can be immediately applied to the $2\times 2\times 2$ IC viewed as a concatenation of interference channels. The approach is first explored by Simeone et. al. \cite{Simeone_mesh}, who view the first hop as an interference channel,  and apply the Han-Kobayashi scheme to split each message into private and common parts. This opens the door to cooperation between the relay nodes  in the second hop based on the shared common message knowledge. Simeone et. al. \cite{Simeone_mesh} explore the gains from coherent combining of the common message signals in the second hop. The limited cooperation between relay nodes is explored as a distributed multiple input single output (MISO) broadcast channel by Thejasvi et. al. in  \cite{Thejasvi_Calderbank_Cochran}. Recently, Cao and Chen \cite{Cao_Chen} explore the $2\times 2\times 2$ IC,  considering each hop as an interference channel, with the added prospect of switching the roles of the relays to, e.g., convert strong interference channels to weak interference channels, which is shown to have significant benefits in achievable data rates in certain SNR regimes.

\subsection*{Interference Alignment Approach}
Any approach that treats either hop (or both hops) of the $2\times 2\times 2$ IC  as an interference channel can only achieve a maximum of 1 DoF, because of the bottleneck created by the 2 user interference channel which has only 1 DoF \cite{Nosratinia_Madsen}. Interestingly, the interference channel approach is highly suboptimal at high SNR. This is because the interference channel approach precludes interference alignment.

Interference alignment refers to a consolidation of undesired signals into smaller dimensions so that the signaling dimensions available for desired signals at each receiver are maximized. Interference alignment was observed first by Birk and Kol \cite{Birk_Kol} for the index coding problem, and then by Maddah-Ali et. al. for the X channel in \cite{MMK}, followed by Weingarten et. al. for the compound vector broadcast channel in \cite{Weingarten_Shamai_Kramer}. The idea was crystallized as a \emph{general} concept  in \cite{Jafar_Shamai, Cadambe_Jafar_int} by Jafar and Shamai, and Cadambe and Jafar, respectively, and has since been applied in increasingly sophisticated forms \cite{Bresler_Parekh_Tse, Nazer_Gastpar_Jafar_Vishwanath, Cadambe_Jafar_Wang, Etkin_Ordentlich, Cadambe_Jafar_Shamai, Sridharan_Jafarian_Vishwanath_Jafar, Motahari_Gharan_Maddah_Khandani, Jafar_corr, Maddah_Tse, Cadambe_Jafar_Maleki} across a variety of communication networks --- both wired and wireless -- often leading to surprising new insights.

Unlike the interference channel approach which can achieve no more than 1 DoF, Cadambe and Jafar show in \cite{Cadambe_Jafar_X} that the $2\times 2\times 2$ IC can achieve $\frac{4}{3}$ DoF almost surely. This is accomplished by a decode and forward approach that treats each hop as an X channel. Specifically, each transmitter divides its message into two independent parts, one intended for each relay. This creates a total of $4$ messages over the first hop, one from each source to each relay node, i.e.,  the $2\times 2$ X channel setting. After decoding the messages from each transmitter, each relay has a message for each destination node, which places the second hop into the X channel setting as well. It is known that the $2\times 2$ X channel with single antenna nodes has $\frac{4}{3}$ DoF. The result was shown first by Jafar and Shamai in \cite{Jafar_Shamai} under the assumption that the channel coefficients are time-varying. By using  a combination of linear beamforming, symbol extensions and asymmetric complex signaling,  Cadambe et. al. showed in  \cite{Cadambe_Jafar_Wang}  that $\frac{4}{3}$ DoF are achievable on the $2\times 2$ X channel even if the channels are held constant for almost all values of channel coefficients. Motahari et. al. \cite{Motahari_Gharan_Khandani_real} proposed the framework of rational dimensions which allows $\frac{4}{3}$ DoF to be achieved almost surely even if the channels are fixed and restricted to real values. Thus, regardless of whether the channels are time-varying or constant and whether they can take complex or only real values, interference alignment through the X channel approach allows the $2\times 2\times 2$ IC to achieve $\frac{4}{3}$ DoF for almost all channel coefficient values.

{\it Remark:} \emph{$\frac{4}{3}$ is the highest achievable DoF result known so far for the $2\times 2\times 2$ IC that is applicable to almost all channel coefficient values.}

\subsection*{Interference Neutralization Approach}
Interference neutralization refers to  the distributed zero forcing of interference when the interfering signal passes through multiple nodes before arriving at the undesired destination. While the terminology \emph{interference neutralization} is more recent \cite{Mohajer_Diggavi_Fragouli_Tse, Mohajer_Tse_Diggavi}, the same essential idea has been around for many years, known by other names such as \emph{distributed orthogonalization}, \emph{distributed zero-forcing}, \emph{multiuser zero-forcing} and  \emph{orthogonalize-and-forward} (see e.g., \cite{Wittneben_mag, Gomadam_Jafar_corr}). A fundamental question for interference neutralization is the minimum number of relays necessary to eliminate all interference. Rankov and Wittneben show in \cite{Rankov_Wittneben} that for the $K\times R\times K$ interference network, a necessary condition for interference neutralization is that $R\geq K(K-1)+1$ relays. Thus, with $2$ sources and $2$ destination nodes, a minimum of $3$ relay nodes is needed for interference neutralization. However, our $2\times 2\times 2$ IC has only $2$ relays, making interference alignment apparently infeasible. 

\subsubsection*{Partially Connected Setting} If the channel over each hop is not fully connected, perfect interference neutralization \emph{is} possible with only $2$ relays. This setting is explored by  Mohajer et. al. in  \cite{Mohajer_Diggavi_Fragouli_Tse,Mohajer_Tse_Diggavi} starting from the deterministic channel setting which leads to a sophisticated re-interpretation of interference neutralization over lattice codes, and ultimately a constant bit gap capacity characterization. In this work, however, we are interested in cases where each hop is fully connected.

\subsubsection*{Opportunistic Interference Neutralization}

The idea of \emph{opportunistic} interference neutralization is introduced by Jeon et. al. in \cite{Jeon_Chung, Jeon_Chung_Jafar}. For a broad class of channel distributions, which includes the commonly studied i.i.d. Rayleigh fading setting, Jeon et. al. show that the DoF achieved correspond to the network min-cut, e.g., for the $2\times 2\times 2$ IC, the DoF = 2, i.e., interference-free transmission is possible without any DoF penalty. This is especially remarkable because no more than half the network min-cut is achievable in a single hop interference network.

Opportunistic interference neutralization is easily understood as follows. For the $2\times 2\times 2$ IC of Fig. \ref{fig:222} consider the setting where the product of the channel matrices $F_{2\times 2}\times G_{2\times 2}$ is a diagonal matrix. Clearly in this case if each relay simply forwards its received signal, the effective end to end channel matrix is a diagonal matrix, i.e, the interference-carrying channel coefficients are reduced to zero, creating a non-interfering channel from each source to its destination. Surprisingly, in this case, even though the channel matrix over each hop may be fully connected the network DoF =2, i.e., the min-cut is achieved. Channel matrices $F$ and $G$ are called complementary matrices if their product is a diagonal matrix with non-zero elements. The idea of opportunistic interference neutralization is to permute the scheduling time instants where the signals from the first hop are transmitted over the second hop, such that the channel matrices of the two hops are complementary. Specifically the relays buffer the received signals and forward a particular received signal only when the second hop channel realization is in a complementary state.

Opportunistic interference neutralization is similar to the ergodic interference alignment scheme of Nazer et. al. \cite{Nazer_Gastpar_Jafar_Vishwanath} for the single hop $K$ user interference channel in the sense that both involve very simple coding and that both are restrictive in their reliance on opportunistic matching of complementary states.  For general channel distributions, even with a large number of parallel channel states, complementary states may not even exist, much less be available in equal proportions to allow one to one matching. The challenge is further compounded if the channel states are held constant. For single hop interference channels it is known that the DoF achieved with ergodic alignment are also achievable for generic time-varying channels through the asymptotic alignment scheme of Cadambe and Jafar  \cite{Cadambe_Jafar_int}. As explained in \cite{SPCOM_Plenary}, even for constant channels (with real or complex channel coefficients) the asymptotic alignment scheme of Cadambe and Jafar can be applied within the rational dimensions framework of Motahari et. al. \cite{Motahari_Gharan_Maddah_Khandani} to achieve the same DoF. Thus, the same $K/2$ DoF are achieved almost surely for all continuous  channel distributions, i.e., DoF are not sensitive to the symmetries of the channel distribution assumed for ergodic interference alignment. However, for multihop interference networks it is not known if the DoF outer bound corresponding to the network min-cut, is achievable for all continuous distributions without relying on the symmetries that allow ergodic pairing of complementary channel states. In particular, for constant channels it is not known if the min-cut is achievable for almost all values of channel coefficients. This brings us to the central question motivating this work.

\subsection*{Main Question - Is the Network DoF Min-Cut Achievable for Generic Channels?}
Consider the $2\times 2\times 2$ IC with generic channel coefficients. The network DoF min-cut value is $2$. If this is to be achieved,  evidently we must have interference-free transmission from each source to its destination, i.e., we need interference neutralization. This raises a significant challenge -- as mentioned in the previous section we need to have $3$ relays to achieve interference neutralization and we have only $2$.  For instance, if the relays $R_1, R_2$ amplify and forward their received signals,  using amplification factors $\alpha_1, \alpha_2$ respectively, then for interference neutralization one should have:
\begin{eqnarray}
F_{11}\alpha_1G_{21}+F_{21}\alpha_2G_{22}=0 &\Rightarrow&\alpha_1/\alpha_2=-F_{21}G_{22}/F_{11}G_{21}\label{eq:neu1}\\
F_{22}\alpha_2G_{12}+F_{12}\alpha_1G_{11}=0&\Rightarrow&\alpha_1/\alpha_2=-F_{22}G_{12}/F_{12}G_{11}\label{eq:neu2}
\end{eqnarray}
Clearly, for generic channels $F_{21}G_{22}/F_{11}G_{21}\neq F_{22}G_{12}/F_{12}G_{11}$, almost surely, so that equations $(\ref{eq:neu1}),(\ref{eq:neu2})$ cannot be simultaneously satisfied.

\subsection*{Main Result - Network DoF Min-Cut is Achievable for Generic Channels}
The main result of this paper is that the min-cut outer bound 2 DoF can be achieved for generic channel coefficients, which is presented in the following theorem.
\begin{theorem}
For the $2\times 2\times 2$ IC with time-varying or constant channel coefficients, the total number of DoF is equal to 2, almost surely.
\end{theorem}

\begin{figure}[!t]
\centering
\includegraphics[width=4.5in]{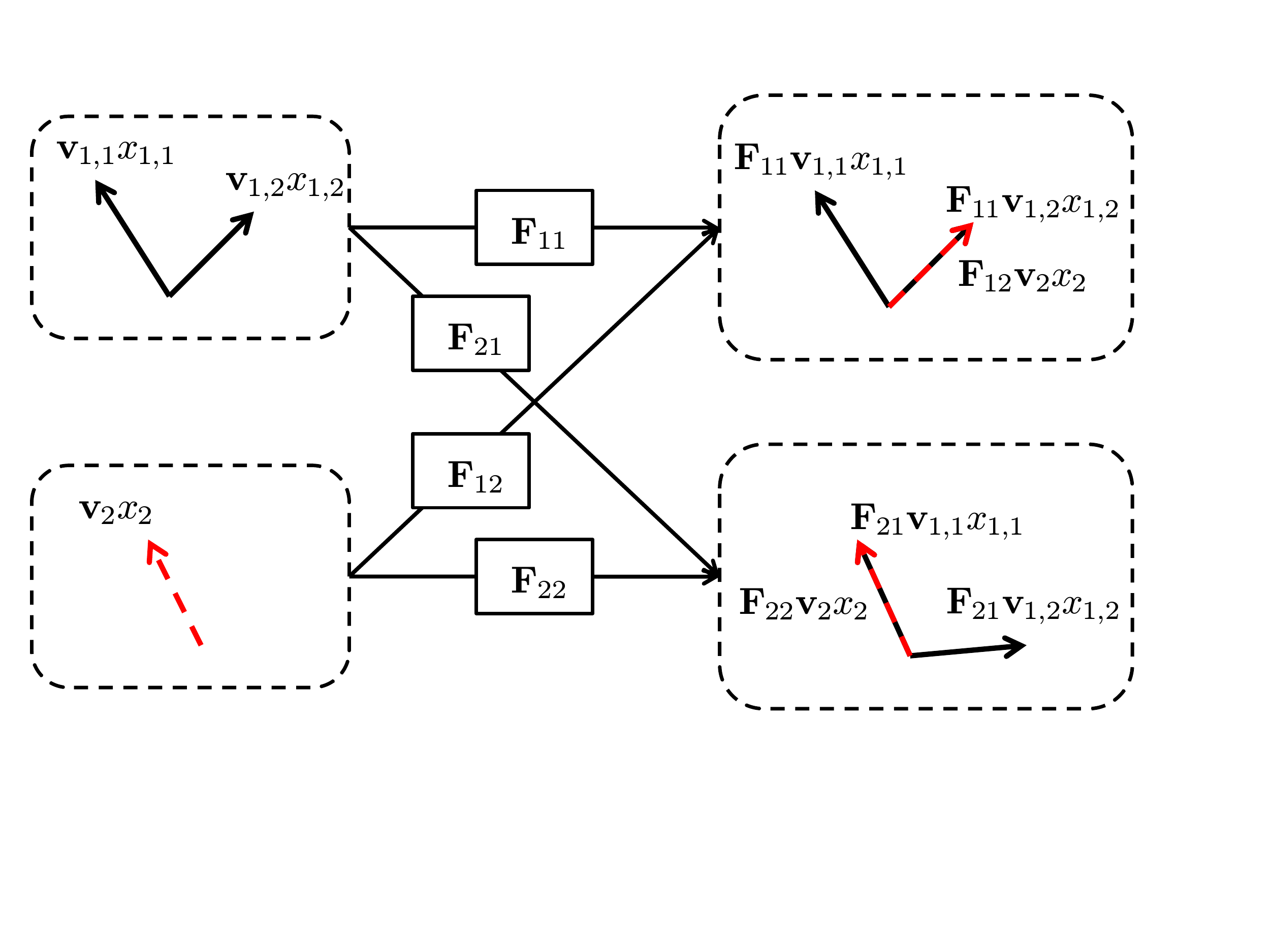}
\caption{Alignment at relays for $M=2$}
\label{fig:alignrelay}
\end{figure}

Like the DoF results for single hop interference networks \cite{Cadambe_Jafar_int}, the problem needs to be solved in high dimensions and in an asymptotic sense. Specifically, we show that with $M$ dimensions, whether it is time, frequency, or rational dimensions, message $W_1$ can access $M$ interference free dimensions while $W_2$ can access $M-1$ interference free dimensions. Thus, $\frac{2M-1}{M}$ DoF can be achieved. Since $M$ can be chosen arbitrarily large, the achieved number of DoF is arbitrarily close to the min-cut bound of $2$ for generic channel coefficients. In other words, almost perfect interference neutralization is achieved asymptotically. The key to this result is a new idea called \emph{aligned interference neutralization} which combines the idea of interference alignment and interference neutralization. While we are primarily interested in the constant channel setting, for ease of exposition we will explain this idea for the case $M=2$ and with time-varying channel coefficients where $\frac{3}{2}$ DoF can be achieved. Ultimately the same time-varying solution will be translated to the case of constant channel coefficients using the rational dimensions framework of \cite{Motahari_Gharan_Maddah_Khandani}. 

Consider a $M=2$ symbol extension of the original network. Then the channel becomes a $2\times 2$ diagonal matrix with distinct diagonal entries. We will show that $W_1$ can achieve 2 DoF while $W_2$ can achieve 1 DoF for a total of  $2M-1=3$ DoF. Source node $S_1$ sends two independent symbols, $x_{1,1}$ and $x_{1,2}$ along beamforming vectors $\mathbf{v}_{1,1}$ and $\mathbf{v}_{1,2}$, respectively. Similarly, source node $S_2$ sends one symbol $x_2$ along beamforming vector $\mathbf{v}_2$. As shown in Fig. \ref{fig:alignrelay}, we design beamforming vectors such that after going through their respective channels, $\mathbf{v}_{1,2}$ and $\mathbf{v}_2$ are along the same direction at relay $R_1$ while $\mathbf{v}_{1,1}$ and $\mathbf{v}_2$ are along the same direction at relay $R_2$, i.e.,
\begin{eqnarray}
\mathbf{F}_{11}\mathbf{v}_{1,2}&=&\mathbf{F}_{12}\mathbf{v}_{2}\\
\mathbf{F}_{21}\mathbf{v}_{1,1}&=&\mathbf{F}_{22}\mathbf{v}_2
\end{eqnarray}
Note that $\mathbf{v}_2$ can be chosen randomly and $\mathbf{v}_{1,1}$ and $\mathbf{v}_{1,2}$ can be solved according to above equations.

\begin{figure}[!t]
\centering
\includegraphics[width=4.5in]{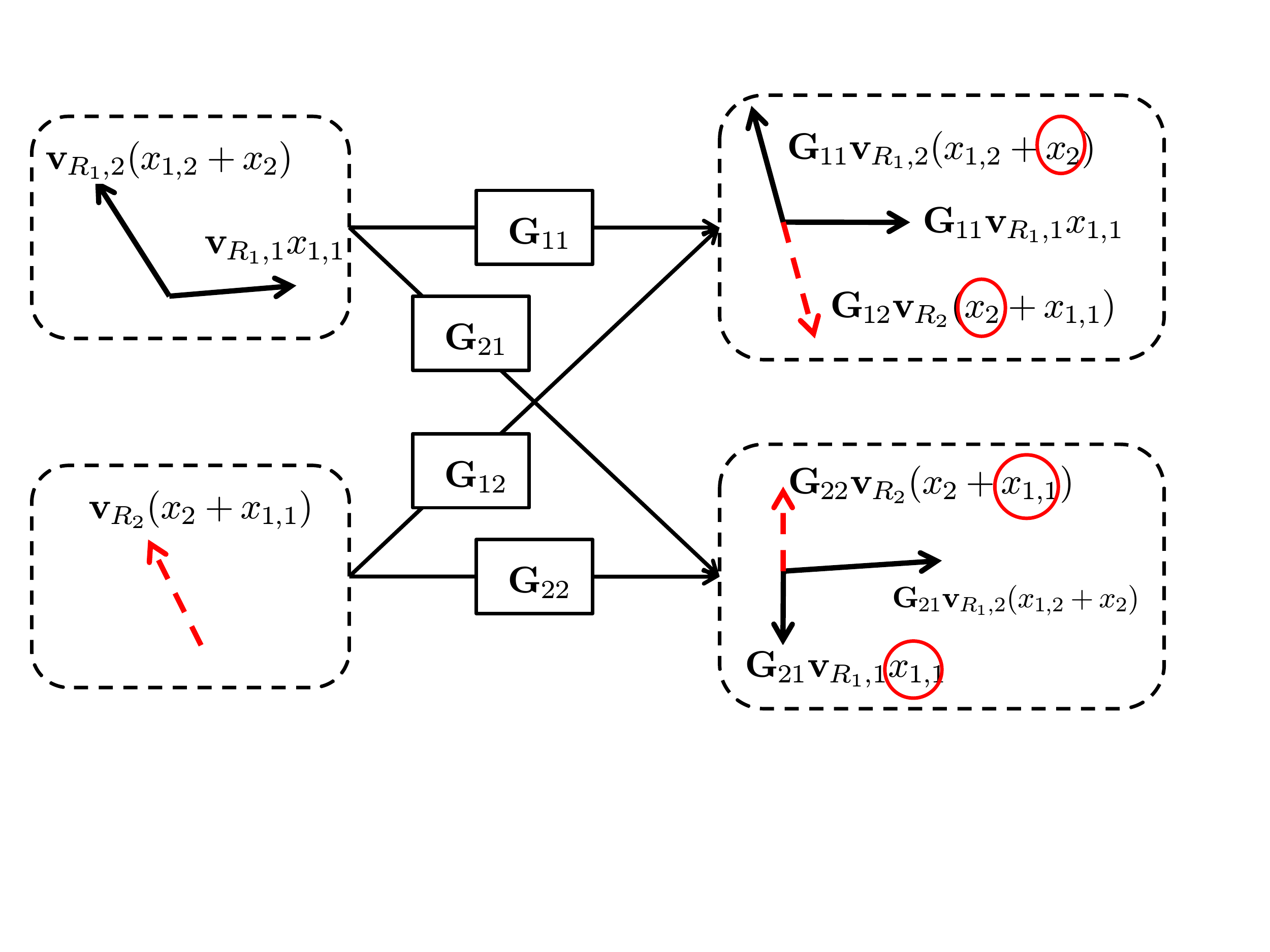}
\caption{Aligned interference neutralization at destinations}
\label{fig:aligndestination}
\end{figure}
After alignment,  $x_{1,1}$ and $x_{1,2}+x_2$ can be isolated through a simple channel matrix inversion operation at $R_1$ while $x_{1,1}+x_2$ and $x_{1,2}$ can be isolated at $R_2$. Then relay $R_1$ sends $x_{1,1}$ and $x_{1,2}+x_2$ in the presence of noise with beamforming vectors $\mathbf{v}_{R_1,1}$ and $\mathbf{v}_{R_1,2}$, respectively. Similarly, relay $R_2$ sends $x_2+x_{1,1}$ in the presence of noise along beamforming vector $\mathbf{v}_{R_2}$. As shown in Fig. \ref{fig:aligndestination}, to neutralize interference $x_2$ at destination $D_1$, we choose
\begin{eqnarray}
\mathbf{G}_{11}\mathbf{v}_{R_1,2}=-\mathbf{G}_{12}\mathbf{v}_{R_2}
\end{eqnarray}
Similarly, to neutralize interference $x_{1,1}$ at destination $D_2$, we choose
\begin{eqnarray}
\mathbf{G}_{21}\mathbf{v}_{R_1,1}=-\mathbf{G}_{22}\mathbf{v}_{R_2}
\end{eqnarray}
Again, $\mathbf{v}_{R_2}$ can be chosen randomly, then $\mathbf{v}_{R_1,1}$ and $\mathbf{v}_{R_1,2}$ can be calculated. After interference neutralization, $D_1$ sees $x_{1,2}-x_{1,1}$ along direction $\mathbf{G}_{11}\mathbf{v}_{R_1,2}$ and $x_{1,1}$ along $\mathbf{G}_{11}\mathbf{v}_{R_1,1}$. Therefore, it can decode $x_{1,1}$ first, and then $x_{1,2}$ to achieve 2 DoF. At $D_2$, $x_2$ is received interference free along $\mathbf{G}_{22}\mathbf{v}_{R_2}$ and it can be decoded by discarding the dimension $\mathbf{G}_{21}\mathbf{v}_{R_1,2}$ along which interference is received.

\section{Channel Model}
The $2 \times 2 \times 2$ IC as shown in Fig. \ref{fig:222} is comprised of two sources, two relays and two destinations. Each source node has a message for its respective destination. In the first hop, the received signal at relay $R_k$, $k\in\{1, 2\}$ in time slot $t$ is
\begin{eqnarray}\label{eq:firsthop}
Y_{R_k}(t)=F_{k1}(t)X_1(t)+F_{k2}(t)X_2(t)+Z_{k}(t)
\end{eqnarray}
where $F_{kj}(t)$, $\forall k, j\in\{1,2\}$, is the complex channel coefficient from source $S_j$ to relay $R_k$, $X_j(t)$ is the input signal from  $S_j$, $Y_{R_k}(t)$ is the received signal at relay $R_k$ and $Z_{k}(t)$ is the independent identically distributed (i.i.d.)  zero mean unit variance circularly symmetric complex Gaussian noise. In the second hop, the received signal at destination $D_k$ in time slot $t$ is given by
\begin{eqnarray}\label{eq:secondhop}
Y_{k}(t)=G_{k1}(t)X_{R_1}(t)+G_{k2}(t)X_{R_2}(t)+N_k(t)
\end{eqnarray}
where $G_{kj}(t)$, $\forall k, j\in\{1,2\}$, is the complex channel coefficient from relay $R_j$ to destination $D_k$, $X_{R_j}(t)$ is the input signal from relay $R_j$, $Y_{k}(t)$ is the received signal at $D_k$ and $N_{k}(t)$ is the i.i.d. zero mean unit variance circularly symmetric complex Gaussian noise. We assume every node in the network has an average power constraint $P$. The relays are full-duplex. We assume that source nodes only know the channels in the first hop, relays know channels in both hops and destination nodes know channels in the second hop.  To avoid degenerate conditions, we assume the absolute values of all channel coefficients are bounded between a nonzero minimum value and a finite maximum value. We will consider two settings where channel coefficients are time-varying or constant.
\begin{enumerate}
\item Channels $F_{kj}(t)$ and $G_{kj}(t)$ are time varying, i.e., the channel coefficients change and are drawn i.i.d. from a continuous distribution for every channel use.
\item Channels $F_{kj}(t)$ and $G_{kj}(t)$ are constant, i.e., the channel coefficients are drawn i.i.d. from a continuous distribution before the transmissions. Once they are drawn, they remain unchanged during the entire transmission. In this case, we will omit the time index for simplicity.
\end{enumerate}

As shown in Fig. \ref{fig:222}, there are two messages in the network. Source $S_k$, $k\in \{1,2\}$ has a message $W_k$ for destination $D_k$. We denote the size of message $W_k$ as $|W_k|$. For the codewords spanning $n$ channel uses, the rates $\mathcal{R}_k=\frac{\log(|W_k|)}{n}$ are achievable if the probability of error for both messages can be simultaneously made arbitrarily small by choosing an appropriately large $n$. The sum-capacity $\mathcal{C}_{\Sigma}(P)$ is the maximum achievable sum rate. The number of degrees of freedom is defined as
\begin{eqnarray}
d=\lim_{P\rightarrow \infty}\frac{\mathcal{C}_{\Sigma}(P)}{\log P}
\end{eqnarray}

\section{Aligned Interference Neutralization}

\subsection{Time-varying channel coefficients - linear scheme}\label{sec:Timevarying}
In this section, we consider the case when the channel coefficients are time-varying. We will show that over $M$ symbol extensions of the original channel, $W_1$ can achieve $M$ DoF while $W_2$ can achieve $M-1$ DoF for a total of $2M-1$ DoF. Thus, the normalized DoF are $\frac{2M-1}{M}$. As $M\rightarrow \infty$, 2 DoF can be achieved almost surely. With $M$ symbol extensions, we effectively have an $M\times M$ MIMO channel with diagonal channel matrices of distinct diagonal entries. Specifically, the channel input-output relations become
\begin{eqnarray}
\mathbf{Y}_{R_k}(t)&=&\mathbf{F}_{k1}(t)\mathbf{X}_1(t)+\mathbf{F}_{k2}(t)\mathbf{X}_2(t)+\mathbf{Z}_{k}(t)\\
\mathbf{Y}_{k}(t)&=&\mathbf{G}_{k1}(t)\mathbf{X}_{R_1}(t)+\mathbf{G}_{k2}(t)\mathbf{X}_{R_2}(t)+\mathbf{N}_k(t)
\end{eqnarray}
where
\begin{eqnarray*}
\mathbf{F}_{kj}(t)&=&\left[\begin{array}{cccc}F_{kj}(Mt+1)& 0 & \cdots & 0 \\ 0 & F_{kj}(Mt+2) & \cdots & 0 \\ \vdots & \cdots& \ddots& \vdots \\ 0 & 0 & \cdots & F_{kj}(Mt+M) \end{array}\right]\\
\mathbf{G}_{kj}(t)&=&\left[\begin{array}{cccc}G_{kj}(Mt+1)& 0 & \cdots & 0 \\ 0 & G_{kj}(Mt+2) & \cdots & 0 \\ \vdots & \cdots& \ddots& \vdots \\ 0 & 0 & \cdots & G_{kj}(Mt+M) \end{array}\right]~~~~~~~~\forall k,j \in\{1,2\}
\end{eqnarray*}
and $\mathbf{X}$,  $\mathbf{Y}$, $\mathbf{Z}$ and $\mathbf{N}$ are $M\times 1$ vectors representing $M$ symbol extensions of $X$, $Y$, $Z$ and $N$, respectively. In the following, we will omit the time index for simplicity.

\subsubsection*{Sources:}
At source node $S_1$, message $W_1$ is split into $M$ sub-messages. Sub-message $W_{1,k_1}$, $k_1\in\{1,\ldots,M\}$, is encoded using a Gaussian codebook with rate equal to 1 DoF and codeword of length $n$ denoted as $x_{1,k_1}(1),\cdots,x_{1,k_1}(n)$. $S_1$ sends symbol $x_{1,k_1}$ along beamforming vector $\mathbf{v}_{1,k_1}$. Then the transmitted signal $\mathbf{X}_1$ is
\begin{eqnarray*}
\mathbf{X}_1=\sum_{k_1=1}^{M}\mathbf{v}_{1,k_1}x_{1,k_1}
\end{eqnarray*}
Similarly, at source node $S_2$, message $W_2$ is split into $M-1$ sub-messages. Sub-message $W_{2,k_2}$, $ k_2\in\{1,\ldots,M-1\}$, is encoded using a Gaussian codebook with rate equal to 1 DoF and codeword of length $n$ denoted as $x_{2,k_2}(1),\cdots, x_{2,k_2}(n)$. $S_2$ sends symbol $x_{2,k_2}$ along beamforming vector $\mathbf{v}_{2,k_2}$. Then the transmitted signal $\mathbf{X}_2$ is
\begin{eqnarray*}
\mathbf{X}_2=\sum_{k_2=1}^{M-1}\mathbf{v}_{2,k_2}x_{2,k_2}
\end{eqnarray*}

\begin{figure}[!t]
\centering
\includegraphics[width=5in]{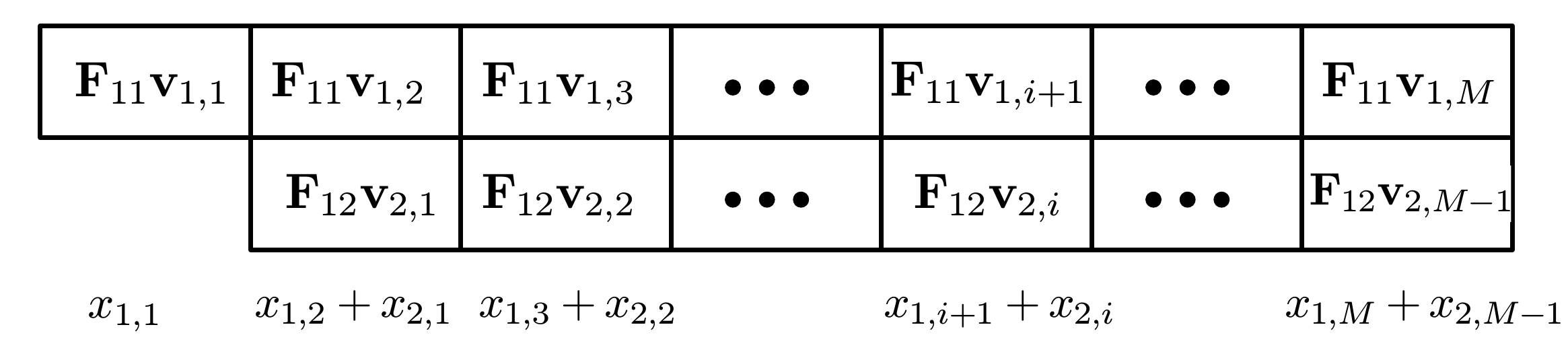}
\caption{Alignment and signals (ignoring noise) in each dimension at $R_1$}
\label{fig:alignr1}
\end{figure}

\begin{figure}[!t]
\centering
\includegraphics[width=5in]{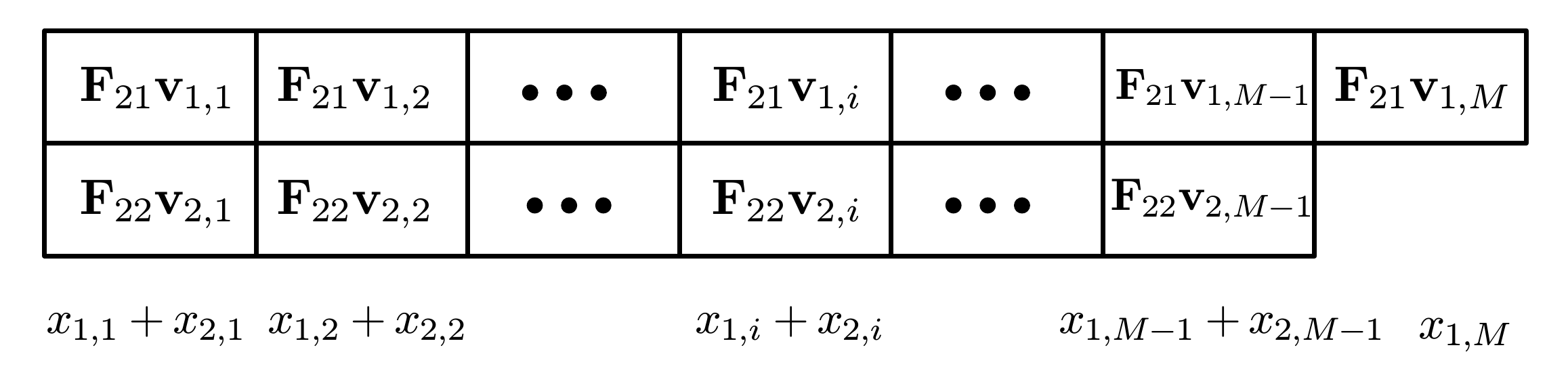}
\caption{Alignment and signals (ignoring noise) in each dimension at $R_2$}
\label{fig:alignr2}
\end{figure}

We will design beamforming vectors $\mathbf{v}_{1,k_1}$ and $\mathbf{v}_{2,k_2}$ such that they align at relays. As shown in Fig. \ref{fig:alignr1}, at relay $R_1$, $\mathbf{F}_{11}\mathbf{v}_{1,i+1}$ aligns with $\mathbf{F}_{12}\mathbf{v}_{2,i}$, $\forall i\in\{1,\ldots,M-1\}$, i.e.,
\begin{eqnarray}
\mathbf{F}_{11}\mathbf{v}_{1,i+1}&=&\mathbf{F}_{12}\mathbf{v}_{2,i}\\
\Rightarrow \mathbf{v}_{1,i+1}&=&\mathbf{F}_{11}^{-1}\mathbf{F}_{12}\mathbf{v}_{2,i}\label{eq:alignr1}.
\end{eqnarray}
At $R_2$, as shown in Fig. \ref{fig:alignr2}, $\mathbf{F}_{21}\mathbf{v}_{1,i}$ aligns with $\mathbf{F}_{22}\mathbf{v}_{2,i}$, $\forall i\in\{1,\ldots,M-1\}$, i.e.,
\begin{eqnarray}
\mathbf{F}_{21}\mathbf{v}_{1,i}&=&\mathbf{F}_{22}\mathbf{v}_{2,i}\\
\Rightarrow \mathbf{v}_{2,i}&=& \mathbf{F}_{22}^{-1}\mathbf{F}_{21}\mathbf{v}_{1,i}\label{eq:alignr2}.
\end{eqnarray}
From \eqref{eq:alignr1} and \eqref{eq:alignr2}, we can draw the dependence of all vectors as shown in Fig. \ref{fig:vectors}. From Fig. \ref{fig:vectors}, it follows that
\begin{eqnarray}
\mathbf{v}_{1,i+1}&=&\left(\mathbf{F}_{11}^{-1}\mathbf{F}_{12}\mathbf{F}_{22}^{-1}\mathbf{F}_{21}\right)^i\mathbf{v}_{1,1}\label{eq:v1}\\
\mathbf{v}_{2,i}&=& \left(\mathbf{F}_{22}^{-1}\mathbf{F}_{21}\mathbf{F}_{11}^{-1}\mathbf{F}_{12}\right)^{i-1}\mathbf{F}_{22}^{-1}\mathbf{F}_{21}\mathbf{v}_{1,1}\label{eq:v2}.
\end{eqnarray}
Note that once $\mathbf{v}_{1,1}$ is determined, then all other vectors can be calculated through \eqref{eq:v1} and \eqref{eq:v2}. We choose $\mathbf{v}_{1,1}$ as an $M\times 1$ vector with all elements equal to one, i.e., $\mathbf{v}_{1,1}=[1~1~\cdots1]^T$. Let $\mathbf{A}=\mathbf{F}_{11}^{-1}\mathbf{F}_{12}\mathbf{F}_{22}^{-1}\mathbf{F}_{21}$. Then $\mathbf{A}$ is a diagonal channel with the $m$th diagonal entry denoted as $A_m$.

\begin{figure}[!t]
\centering
\includegraphics[width=6in]{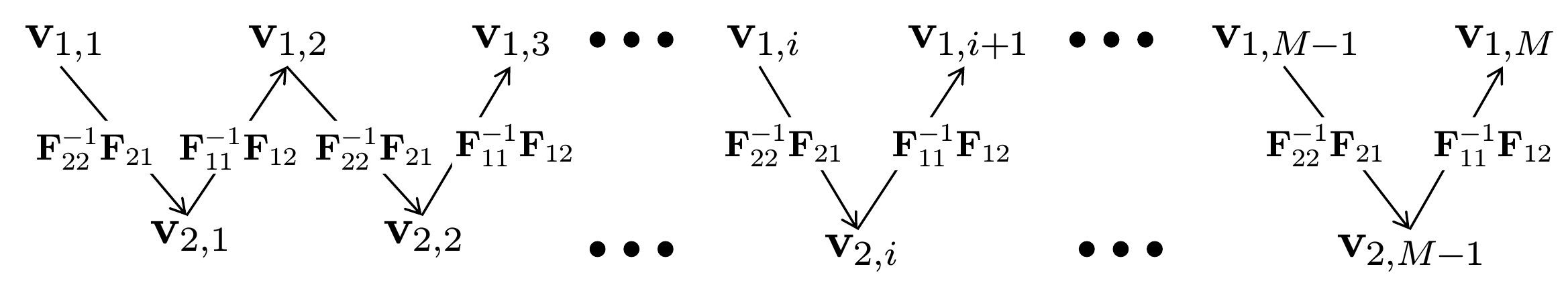}
\caption{Dependence of vectors for the first hop}
\label{fig:vectors}
\end{figure}

Next, we prove that $\mathbf{v}_{1,1}, \mathbf{v}_{1,2}, \ldots, \mathbf{v}_{1,M}$ are linearly independent. From \eqref{eq:v1} and $\mathbf{v}_{1,1}=[1~\cdots~1]^T$, we have
\begin{eqnarray}
\mathbf{v}_{1,i+1}=\mathbf{A}^i \left[\begin{array}{c} 1 \\ \vdots \\ 1 \end{array}\right]=\left[\begin{array}{c}A_{1}^i\\ \vdots \\ A_{M}^i \end{array}\right]~~~i\in\{1,\cdots, M-1\}.
\end{eqnarray}
Let $\mathbf{B}=\left[\mathbf{v}_{1,1}, \mathbf{v}_{1,2}, \ldots, \mathbf{v}_{1,M}\right]$. Then
\begin{eqnarray}
\mathbf{B}=\left[\begin{array}{cccc}1&A_{1}&\cdots& A_{1}^{M-1}\\ \vdots & \vdots & \ddots & \vdots \\ 1& A_{M}&\cdots& A_{M}^{M-1}\end{array}\right].
\end{eqnarray}
Notice that $\mathbf{B}$ is a Vandermonde matrix and its determinant is
\begin{eqnarray}
\det(\mathbf{B})=\prod_{1\leq i<j \leq M}(A_{j}-A_{i}).
\end{eqnarray}
Since all channel coefficients are time-varying and drawn i.i.d. according to a continuous distribution, $A_m$ are all distinct almost surely. Therefore, $\det(\mathbf{B})$ is not equal to zero almost surely, which establishes the linear independence of $\mathbf{v}_{1,1}, \mathbf{v}_{1,2}, \ldots, \mathbf{v}_{1,M}$. Similarly, $\mathbf{v}_{2,1},\cdots,\mathbf{v}_{2,M-1}$ are linearly independent.

\subsubsection*{Relays:}
Let us first consider the received signals at relays. After alignment, at $R_1$, the signal (ignoring noise) in the first dimension which is along $\mathbf{F}_{11}\mathbf{v}_{1,1}$ is $x_{1,1}$ and in the $(i+1)$th dimension which is along $\mathbf{F}_{11}\mathbf{v}_{1,i+1}$, $i\in\{1,\cdots, M-1\}$, the signal (ignoring noise) is $x_{1,i+1}+x_{2,i}$. This is shown in Fig. \ref{fig:alignr1}. Specifically, the received signal at $R_1$ is
\begin{eqnarray}
\mathbf{Y}_{R_1}&=&\mathbf{F}_{11}\mathbf{X}_{1}+\mathbf{F}_{12}\mathbf{X}_{2}+\mathbf{Z}_1\notag\\
&=&\mathbf{F}_{11}\sum_{k_1=1}^{M}\mathbf{v}_{1,k_1}x_{1,k_1}+\mathbf{F}_{12}\sum_{k_2=1}^{M-1}\mathbf{v}_{2,k_2}x_{2,k_2}+\mathbf{Z}_1\notag\\
&\stackrel{(a)}{=}&\mathbf{F}_{11}\mathbf{v}_{1,1}x_{1,1}+\sum_{i=1}^{M-1}\mathbf{F}_{11}\mathbf{v}_{1,i+1}(x_{1,i+1}+x_{2,i})+\mathbf{Z}_1
\end{eqnarray}
where $(a)$ uses the alignment condition \eqref{eq:alignr1}.

The relay will amplify and forward its received signal by multiplying a matrix. This is done with two stages. The relay first isolates signals in each dimension by multiplying the inverse of the effective $M\times M$ channel matrix. Mathematically,
\begin{eqnarray}\label{eq:invertr1}
\left[\begin{array}{c}x_{R_1,1}\\ x_{R_1,2}\\ \vdots \\ x_{R_1,i+1}\\ \vdots\\ x_{R_1,M}\end{array}\right]=\mathbf{F}_{R_1}^{-1}\mathbf{Y}_{R_1}
=\left[\begin{array}{c}x_{1,1}\\ x_{1,2}+x_{2,1}\\ \vdots\\ x_{1,i+1}+x_{2,i} \\ \vdots \\ x_{1,M}+x_{2,M-1}\end{array}\right]+\mathbf{F}_{R_1}^{-1}\mathbf{Z}_1~~~~~~i\in\{1,\cdots, M-1\}
\end{eqnarray}
where $\mathbf{F}_{R_1}=[\mathbf{F}_{11}\mathbf{v}_{1,1}~\mathbf{F}_{11}\mathbf{v}_{1,2}~\cdots~\mathbf{F}_{11}\mathbf{v}_{1,M}]$. Then $R_1$ sends $x_{R_1,k_1}$ along beamforming vector $\mathbf{v}_{R_1,k_1}$, $k_1 \in\{1,\cdots, M\}$, i.e.,
\begin{eqnarray*}
\mathbf{X}_{R_1}=\sum_{k_1=1}^{M}\mathbf{v}_{R_1,k_1}x_{R_1,k_1}.
\end{eqnarray*}
Similarly, as shown in Fig. \ref{fig:alignr2}, at relay $R_2$, the signal (ignoring noise) in the $i$th dimension which is along $\mathbf{F}_{21}\mathbf{v}_{1,i}$ is $x_{1,i}+x_{2,i}$, $\forall i\in\{1,\cdots,M-1\}$, and in the $M$th dimension which is along $\mathbf{F}_{21}\mathbf{v}_{1,M}$, the signal is $x_{1,M}$. Mathematically, the received signal at $R_2$ is
\begin{eqnarray}
\mathbf{Y}_{R_2}&=&\mathbf{F}_{21}\mathbf{X}_{1}+\mathbf{F}_{22}\mathbf{X}_{2}+\mathbf{Z}_2\notag\\
&\stackrel{(a)}{=}&\sum_{i=1}^{M-1}\mathbf{F}_{21}\mathbf{v}_{1,i}(x_{1,i}+x_{2,i})+\mathbf{F}_{21}\mathbf{v}_{1,M}x_{1,M}+\mathbf{Z}_2
\end{eqnarray}
where $(a)$ uses the alignment condition \eqref{eq:alignr2}. Then after inverting the effective channel matrix, the signals in each dimension are
\begin{eqnarray}\label{eq:invertr2}
\left[\begin{array}{c}x_{R_2,1}\\ \vdots \\ x_{R_2,i} \\ \vdots \\ x_{R_2,M-1} \\ x_{R_2,M}\end{array}\right]=\mathbf{F}_{R_2}^{-1}\mathbf{Y}_{R_2}
=\left[\begin{array}{c}x_{1,1}+x_{2,1}\\ \vdots\\ x_{1,i}+x_{2,i} \\ \vdots \\ x_{1,M-1}+x_{2,M-1} \\ x_{1,M}\end{array}\right]+\mathbf{F}_{R_2}^{-1}\mathbf{Z}_2 ~~~~~~~~~~ i \in\{1,\cdots, M-1\}
\end{eqnarray}
where $\mathbf{F}_{R_2}=[\mathbf{F}_{21}\mathbf{v}_{1,1}~\mathbf{F}_{21}\mathbf{v}_{1,2}~\cdots~\mathbf{F}_{21}\mathbf{v}_{1,M}]$. Then $R_2$ sends $x_{R_2,k_2}$ along beamforming vector $\mathbf{v}_{R_2,k_2}$, $k_2 \in\{1,\cdots, M-1\}$, i.e.,
\begin{eqnarray*}
\mathbf{X}_{R_2}=\sum_{k_2=1}^{M-1}\mathbf{v}_{R_2,k_2}x_{R_2,k_2}.
\end{eqnarray*}

\begin{figure}[!t]
\centering
\includegraphics[width=6in]{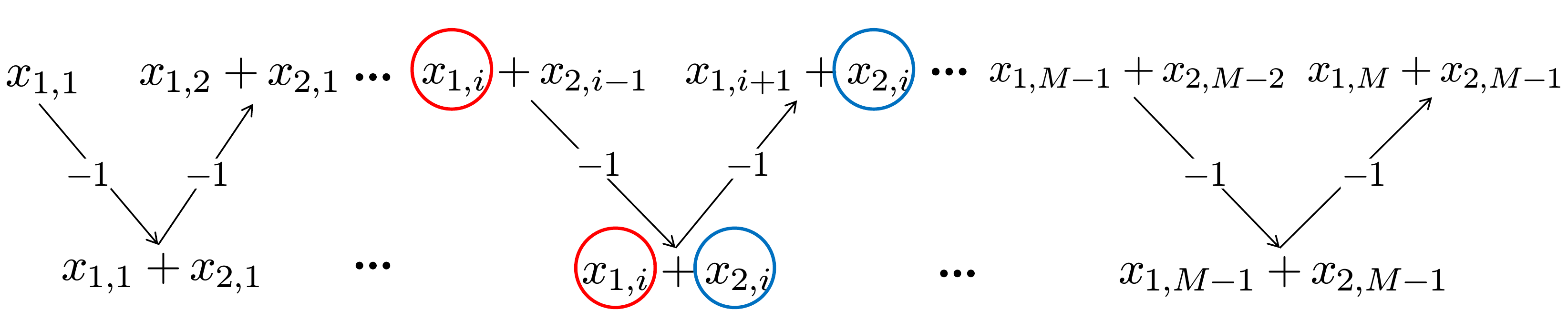}
\caption{Interference neutralization in the second hop}
\label{fig:alignedin}
\end{figure}

\begin{figure}[!t]
\centering
\includegraphics[width=5in]{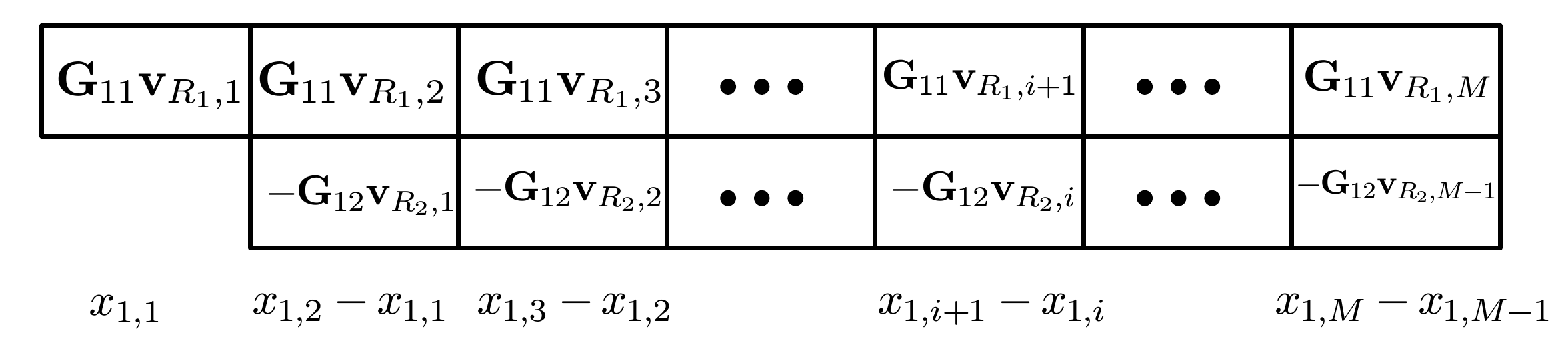}
\caption{Aligned interference neutralization at $D_1$}
\label{fig:alignd1}
\end{figure}

\begin{figure}[!t]
\centering
\includegraphics[width=5in]{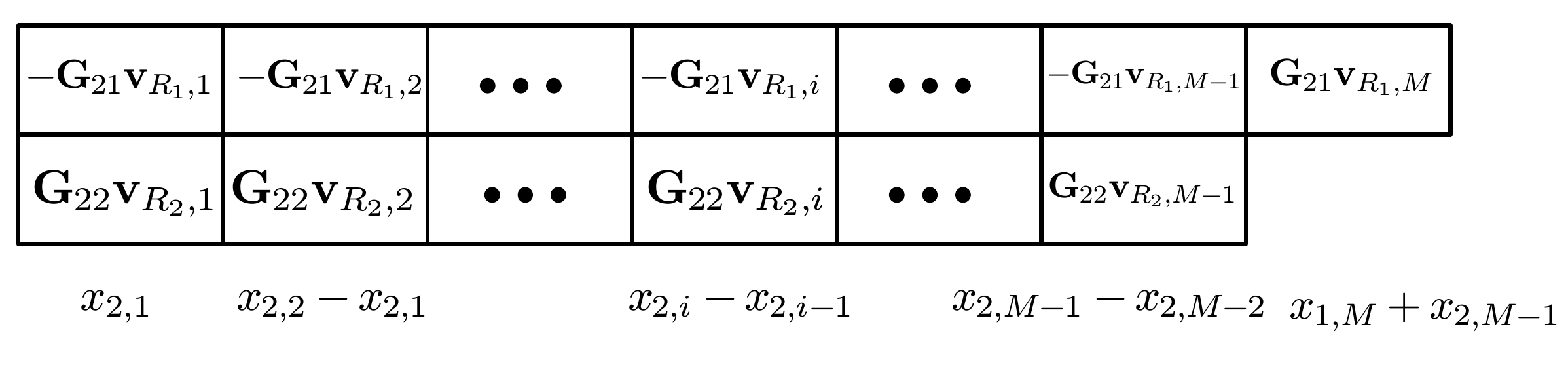}
\caption{Aligned interference neutralization at $D_2$}
\label{fig:alignd2}
\end{figure}
We will design $\mathbf{v}_{R_1,k_1}$ and $\mathbf{v}_{R_2,k_2}$ such that all interference can be cancelled over the air at the destinations. To understand how to neutralize all interference, let us focus on the symbols transmitted through two relays. As shown in Fig. \ref{fig:alignedin}, the symbol that occupies the $i+1$th dimension at $R_1$ is $x_{1,i+1}+x_{2,i}$ and the symbols that occupy the $i$th dimension at $R_2$ is $x_{1,i}+x_{2,i}$. Notice if these two signals are received along the same dimension but with complementary signs at $D_1$, then interference $x_{2,i}$ can be neutralized. This can be done by choosing following alignment condition as shown in Fig. \ref{fig:alignd1}:
\begin{eqnarray}
\mathbf{G}_{11}\mathbf{v}_{R_1,i+1}&=&-\mathbf{G}_{12}\mathbf{v}_{R_2,i}\\
\Rightarrow \mathbf{v}_{R_1,i+1}&=&-\mathbf{G}_{11}^{-1}\mathbf{G}_{12}\mathbf{v}_{R_2,i}~~~~~~~~~~\forall i\in\{1,\ldots,M-1\}.\label{eq:alignd1}
\end{eqnarray}
Similarly, as shown in Fig. \ref{fig:alignedin}, the symbols that occupy the $i$th dimension at $R_1$ and $R_2$ are $x_{1,i}+x_{2,i-1}$ and $x_{1,i}+x_{2,i}$. If these two symbols are received along the same dimension but with complementary signs at $D_2$, then interference $x_{1,i}$ can be neutralized. This can be done by choosing following alignment conditions as shown in Fig. \ref{fig:alignd2}
\begin{eqnarray}
-\mathbf{G}_{21}\mathbf{v}_{R_1,i}&=&\mathbf{G}_{22}\mathbf{v}_{R_2,i}\\
\Rightarrow \mathbf{v}_{R_2,i}&=&-\mathbf{G}_{22}^{-1}\mathbf{G}_{21}\mathbf{v}_{R_1,i}~~~~~~~~~~\forall i\in\{1,\ldots,M-1\}\label{eq:alignd2}.
\end{eqnarray}
The dependence of all vectors is shown in Fig. \ref{fig:vectorsh2}. From \eqref{eq:alignd1} and \eqref{eq:alignd2}, it follows that
\begin{eqnarray}
\mathbf{v}_{R_1,i+1}&=&\left(\mathbf{G}_{11}^{-1}\mathbf{G}_{12}\mathbf{G}_{22}^{-1}\mathbf{G}_{21}\right)^i\mathbf{v}_{R_1,1}\label{eq:vr1}\\
\mathbf{v}_{R_2,i}&=& -\left(\mathbf{G}_{22}^{-1}\mathbf{G}_{21}\mathbf{G}_{11}^{-1}\mathbf{G}_{12}\right)^{i-1}\mathbf{G}_{22}^{-1}\mathbf{G}_{21}\mathbf{v}_{R_1,1}\label{eq:vr2}
\end{eqnarray}
Note that once $\mathbf{v}_{R_1,1}$ is determined, then all other vectors can be calculated through \eqref{eq:vr1} and \eqref{eq:vr2}. Again, we choose $\mathbf{v}_{R_1,1}=[1~\cdots~1]^T$.
\begin{figure}[!t]
\centering
\includegraphics[width=6in]{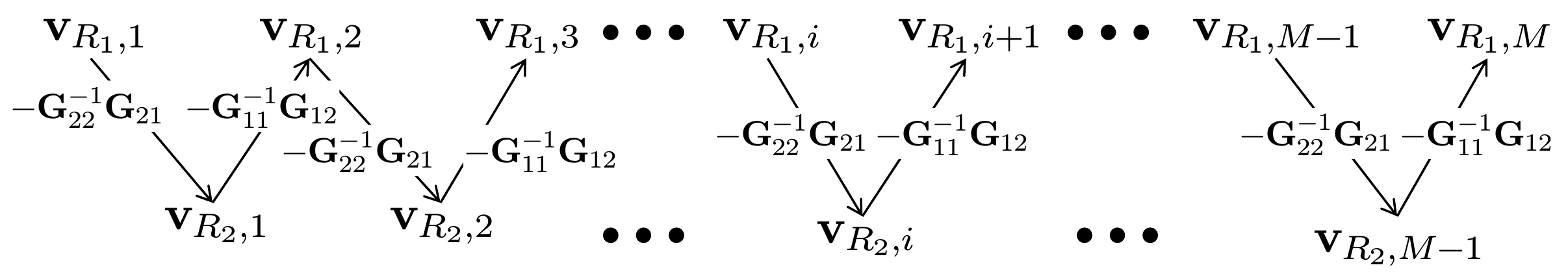}
\caption{Dependence of vectors for the second hop}
\label{fig:vectorsh2}
\end{figure}
Note that \eqref{eq:v1} and \eqref{eq:vr1} are in the same form. Therefore, we can prove that  $\mathbf{v}_{R_1,1},\cdots,\mathbf{v}_{R_1,M}$ are linearly independent. Similarly, $\mathbf{v}_{R_2,1},\cdots,\mathbf{v}_{R_2,M-1}$ are linearly independent.

\subsubsection*{Destinations:}
After aligned interference neutralization, each destination can decode its desired signals. At $D_1$, the received signal is
\begin{eqnarray}
\mathbf{Y}_1&=&\mathbf{G}_{11}\mathbf{X}_{R_1}+\mathbf{G}_{12}\mathbf{X}_{R_2}+\mathbf{N}_1\notag\\
&=&\mathbf{G}_{11}\sum_{k_1=1}^{M}\mathbf{v}_{R_1,k_1}x_{R_1,k_1}+\mathbf{G}_{12}\sum_{k_2=1}^{M-1}\mathbf{v}_{R_2,k_2}x_{R_2,k_2}+\mathbf{N}_1\notag\\
&\stackrel{(a)}{=}&\mathbf{G}_{11}\mathbf{v}_{R_1,1}(x_{1,1}+z'_{1,1})+\sum_{i=1}^{M-1}\mathbf{G}_{11}\mathbf{v}_{R_1,i+1}(x_{1,i+1}-x_{1,i}+z'_{1,i+1}-z'_{2,i})+\mathbf{N}_1
\end{eqnarray}
where $(a)$ uses the alignment condition \eqref{eq:alignd1} and $z'_{1,k_1}$ and $z'_{2,k_2}$ are the $k_1$th and $k_2$th element of the noise vectors $\mathbf{F}_{R_1}^{-1}\mathbf{Z}_1$ and $\mathbf{F}_{R_2}^{-1}\mathbf{Z}_2$ in \eqref{eq:invertr1} and \eqref{eq:invertr2}, respectively. Therefore, as shown in Fig. \ref{fig:alignd1}, the signal (ignoring noise) in the first dimension which is along $\mathbf{G}_{11}\mathbf{v}_{R_1,1}$ is $x_{1,1}$ and in the $(i+1)$th dimension which is along $\mathbf{G}_{11}\mathbf{v}_{R_1,i+1}$, $i\in\{1,\cdots, M-1\}$, the signal is $x_{1,i+1}-x_{1,i}$.  Then receiver $D_1$ can first decode $x_{1,1}$ and subtract it from the second dimension to decode $x_{1,2}$ and so on to decode all desired symbols in a chain.

Similarly, the received signal at $D_2$ is
\begin{eqnarray}
\mathbf{Y}_2&=&\mathbf{G}_{21}\mathbf{X}_{R_1}+\mathbf{G}_{22}\mathbf{X}_{R_2}+\mathbf{N}_2\notag\\
&=&\mathbf{G}_{21}\sum_{k_1=1}^{M}\mathbf{v}_{R_1,k_1}x_{R_1,k_1}+\mathbf{G}_{22}\sum_{k_2=1}^{M-1}\mathbf{v}_{R_2,k_2}x_{R_2,k_2}+\mathbf{N}_2\notag\\
&\stackrel{(a)}{=}&\mathbf{G}_{21}\mathbf{v}_{R_1,M}(x_{1,M}+x_{2,M-1}+z'_{1,M})+\sum_{i=1}^{M-1}\mathbf{G}_{22}\mathbf{v}_{R_2,i}(x_{2,i}-x_{2,i-1}-z'_{1,i}+z'_{2,i})+\mathbf{N}_2
\end{eqnarray}
where $(a)$ uses the alignment condition \eqref{eq:alignd2} and $x_{2,0}=0$. Therefore, as shown in Fig. \ref{fig:alignd2}, the signal (ignoring noise) in the $i$th dimension along $\mathbf{G}_{22}\mathbf{v}_{R_2,i}$ is $x_{2,i}-x_{2,i-1}$, $\forall i\in\{1,\cdots, M-1\}$ and the signal (ignoring noise) in the $M$th dimension along $\mathbf{G}_{21}\mathbf{v}_{R_1,M}$ is $x_{1,M}+x_{2,M-1}$. Then  $D_2$ can first decode $x_{2,1}$ and subtract it from the second dimension to decode $x_{2,2}$ and so on to decode all desired symbols in a chain.

\subsection{Constant channel coefficients - rational dimension framework}
In this section, we will restrict the channels to constant values. The achievable scheme proposed for the time-varying case cannot be applied directly to the constant case. This is because after symbol extension, $\mathbf{F}$ and $\mathbf{G}$ are essentially scaled identity matrices if the channels are constant. From \eqref{eq:v1}, it can be easily seen that $\mathbf{v}_{1,i+1}$, $i\in\{1,\cdots,M-1\}$, are scaling of $\mathbf{v}_{1,1}$ and thus they are linearly dependent. In this case, we will use the framework of rational dimensions introduced in \cite{Motahari_Gharan_Maddah_Khandani} for $K$ user interference channel with constant channel coefficients. With the rational dimension framework, the aligned interference neutralization scheme proposed in the last section can be carried over to the real constant channel where all signals, channel coefficients and noises are real values. In this case, the DoF is defined as $d=\lim_{P\rightarrow \infty}\frac{C_{\Sigma}(P)}{\frac{1}{2}\log P}$. The result can also be generalized to the complex case using Theorem 7 in \cite{Maddah_Compound}.

\subsubsection*{Sources:}
 At source node $S_1$, message $W_1$ is split into $M$ sub-messages. Sub-message $W_{1,k_1}$, $k_1\in\{1,\ldots,M\}$, is encoded using a codebook with the codeword of length $n$ denoted as $x_{1,k_1}(1),\cdots,x_{1,k_1}(n)$. For any
$\epsilon>0$ and a constant $\gamma$, let $\mathcal{C}$ denote all integers in the interval $\left[-\gamma P^{\frac{1-\epsilon}{2(M+\epsilon)}}, \gamma P^{\frac{1-\epsilon}{2(M+\epsilon)}}\right]$, i.e.,
\begin{eqnarray}
\mathcal{C}=\left\{x: x\in\mathbb{Z}\cap \left[-\gamma P^{\frac{1-\epsilon}{2(M+\epsilon)}}, \gamma P^{\frac{1-\epsilon}{2(M+\epsilon)}}\right]\right\}.
\end{eqnarray}
$x_{1,k_1}$ is obtained by uniform i.i.d. sampling on $\mathcal{C}$.  Essentially, each sub-message carries $\frac{1-\epsilon}{M+\epsilon}$ DoF. $S_1$ sends the linear combination of $x_{1,1}, \cdots, x_{1,M}$ with real coefficients $v_{1,1}, \cdots, v_{1,M}$. Then the transmitted signal $X_1$ is
\begin{eqnarray*}
X_1=A\sum_{k_1=1}^{M}v_{1,k_1}x_{1,k_1}
\end{eqnarray*}
where $A$ is a normalizing constant chosen to satisfy the power constraint.

Similarly, at $S_2$, message $W_2$ is split into $M-1$ sub-messages. Sub-message $W_{2,k_2}$, $ k_2\in\{1,\ldots,M-1\}$, is encoded using a codebook with codeword of length $n$ denoted as $x_{2,k_2}(1),\cdots,x_{2,k_2}(n)$ where each symbol is obtained by uniform i.i.d. sampling on $\mathcal{C}$.  Then the transmitted signal $X_2$ is
\begin{eqnarray*}
X_2=A\sum_{k_2=1}^{M-1}v_{2,k_2}x_{2,k_2}
\end{eqnarray*}
The power constraints at both source nodes are
\begin{eqnarray}
E\left[X_1^2\right]&=&A^2\sum_{k_1=1}^Mv^2_{1,k_1}E\left[x_{1,k_1}^2\right]\leq \gamma^2A^2\underbrace{\sum_{k_1=1}^Mv^2_{1,k_1}}_{\xi_1^2} P^{\frac{1-\epsilon}{(M+\epsilon)}}\leq P\\
E\left[X_2^2\right]&=&A^2\sum_{k_2=1}^{M-1}v^2_{2,k_2}E\left[x_{2,k_2}^2\right]\leq \gamma^2A^2\underbrace{\sum_{k_2=1}^{M-1}v^2_{2,k_2}}_{\xi_2^2} P^{\frac{1-\epsilon}{(M+\epsilon)}}\leq P.
\end{eqnarray}
To satisfy  power constraints at both transmitters, we choose
\begin{eqnarray}
A=\frac{\xi}{\gamma}P^{\frac{M-1+2\epsilon}{2(M+\epsilon)}}
\end{eqnarray}
where $\xi=\min(\frac{1}{\xi_1},\frac{1}{\xi_2})$.

Similar to the linear scheme, we choose the following alignment conditions at $R_1$
\begin{eqnarray}\label{eq:alignedr1rational}
F_{11}v_{1,i+1}&=&F_{12}v_{2,i}.
\end{eqnarray}
At $R_2$, we choose the following alignment:
\begin{eqnarray}\label{eq:alignedr2rational}
F_{21}v_{1,i}&=&F_{22}v_{2,i}.
\end{eqnarray}
Then we have
\begin{eqnarray}
v_{1,i+1}&=&\left(F_{11}^{-1}F_{12}F_{22}^{-1}F_{21}\right)^iv_{1,1}\\
v_{2,i}&=& \left(F_{22}^{-1}F_{21}F_{11}^{-1}F_{12}\right)^{i-1}F_{22}^{-1}F_{21}v_{1,1}
\end{eqnarray}
Again, once $v_{1,1}$ is determined, then all other scaling factors can be calculated through above equations. We choose $v_{1,1}=(F_{11}F_{22})^{M-1}$. Thus, we have
\begin{eqnarray}
v_{1,i+1}&=&(F_{12}F_{21})^{i}(F_{11}F_{22})^{M-i-1}\\
v_{2,i}&=&F_{11}^{M-i}F_{12}^{i-1}F_{21}^iF_{22}^{M-1-i}~~~i\in\{1,\cdots, M-1\}.
\end{eqnarray}

\subsubsection*{Relays:}
Instead of amplifying and forward the received signal at the relays as in the linear scheme, the relays will make hard decisions on the signals received in each rational dimension and then forward them. Therefore, unlike the linear scheme, the noise will not be built up at relays. The received signal at relay $R_1$ is
\begin{eqnarray}
Y_{R_1}&=&F_{11}X_{1}+F_{12}X_{2}+Z_1\notag\\
&=&AF_{11}\sum_{k_1=1}^{M}v_{1,k_1}x_{1,k_1}+AF_{12}\sum_{k_2=1}^{M-1}v_{2,k_2}x_{2,k_2}+Z_1\notag\\
&\stackrel{(a)}{=}&AF_{11}v_{1,1}x_{1,1}+\sum_{i=1}^{M-1}AF_{11}v_{1,i+1}\underbrace{(x_{1,i+1}+x_{2,i})}_{x_{R_1,i+1}}+Z_1
\end{eqnarray}
where $(a)$ uses the alignment condition \eqref{eq:alignedr1rational}. Note that $x_{R_1,i+1}$ are sum of two symbols, which is also an integer but in the interval $\left[-2\gamma P^{\frac{1-\epsilon}{2(M+\epsilon)}}, 2\gamma P^{\frac{1-\epsilon}{2(M+\epsilon)}}\right]$. Let $x_{R_1,1}=x_{1,1}$. Therefore, the received signal is a noisy observation of a point from the following constellation:
\begin{eqnarray}
\mathcal{C}_{R_1}=\left\{A\left(F_{11}v_{1,1}x_{R_1,1}+\cdots+F_{11}v_{1,M}x_{R_1,M}\right)\right\}
\end{eqnarray}
Notice that $v_{1,1},\cdots,v_{1,M}$ are distinct monomial functions of channel coefficients and thus rationally independent almost surely. Thus, there is a one-to-one mapping from $\mathcal{C}_{R_1}$ to $x_{R_1,k_1}$, $k_1\in\{1,\cdots,M\}$. Relay $R_1$ will find the point in $\mathcal{C}_{R_1}$ which has the minimal distance between $Y_{R_1}$, and map the point to $\hat{x}_{R_1,k_1}$ to make a hard decision on $x_{R_1,k_1}$. From \cite{Motahari_Gharan_Maddah_Khandani}, it can be shown that the minimum distance between two points in $\mathcal{C}_{R_1}$ increases with $P$ almost surely. Therefore, it can be shown that the error probability of estimating $x_{R_1,k_1}$, $\text{Pr}(\hat{x}_{R_1,k_1}\neq x_{R_1,k_1})$, will go to zero as the power $P$ goes to infinity. Then the transmitted signal at $R_1$ is
\begin{eqnarray*}
X_{R_1}=B\left(\sum_{k_1=1}^{M}v_{R_1,k_1}\hat{x}_{R_1,k_1}\right)
\end{eqnarray*}
where $B$ is a normalizing constant to satisfy the power constraint.

Similarly, at $R_2$, the received signal is
\begin{eqnarray}
Y_{R_2}&=&F_{21}X_{1}+F_{22}X_{2}+Z_2\notag\\
&=&AF_{21}\sum_{k_1=1}^{M}v_{1,k_1}x_{1,k_1}+AF_{22}\sum_{k_2=1}^{M-1}v_{2,k_2}x_{2,k_2}+Z_2\notag\\
&\stackrel{(a)}{=}&AF_{21}v_{1,M}x_{1,M}+A\sum_{i=1}^{M-1}F_{21}v_{1,i}\underbrace{(x_{1,i}+x_{2,i})}_{x_{R_2,i}}+Z_2
\end{eqnarray}
where $(a)$ uses the alignment condition \eqref{eq:alignedr2rational}. Similarly, relay $R_2$ will make a hard decision on $x_{R_2,i}$. Note that $x_{R_2,i}$ are integers in the interval $\left[-2\gamma P^{\frac{1-\epsilon}{2(M+\epsilon)}}, 2\gamma P^{\frac{1-\epsilon}{2(M+\epsilon)}}\right]$. From \cite{Motahari_Gharan_Maddah_Khandani}, it can be shown that the error probability of estimating $x_{R_2,i}$, $\text{Pr}(\hat{x}_{R_2,i}\neq x_{R_2,i})$, will go to zero as the power $P$ goes to infinity.
Then the transmitted signal at $R_2$ is
\begin{eqnarray*}
X_{R_2}=B\sum_{k_2=1}^{M-1}v_{R_2,k_2}\hat{x}_{R_2,k_2}
\end{eqnarray*}
Now consider the power constraints at two relays. At $R_1$
\begin{eqnarray}
E\left[X_{R_1}^2\right]&=&B^2\sum_{m,n=1}^{M} v_{R_1,m}v_{R_1,n}E(\hat{x}_{R_1,m}\hat{x}_{R_1,n})\\
&\leq& 4B^2\gamma^2P^{\frac{1-\epsilon}{(M+\epsilon)}}\underbrace{\sum_{m,n=1}^{M} v_{R_1,m}v_{R_1,n}}_{\xi_1^2}\leq P
\end{eqnarray}
where we use the fact that $\hat{x}_{R_1,k_1}\leq 2\gamma P^{\frac{1-\epsilon}{2(M+\epsilon)}}$.
Similarly, the power constraint at $R_2$
\begin{eqnarray}
E\left[X_{R_2}^2\right]&=&B^2\sum_{m,n=1}^{M-1} v_{R_2,m}v_{R_2,n}E(\hat{x}_{R_2,m}\hat{x}_{R_2,n})\\
&\leq& 4B^2\gamma^2P^{\frac{1-\epsilon}{(M+\epsilon)}}\underbrace{\sum_{m,n=1}^{M-1} v_{R_2,m}v_{R_2,n}}_{\xi_2^2}\leq P
\end{eqnarray}
To satisfy power constraints, we choose
\begin{eqnarray}
B=\frac{\xi}{2\gamma}P^{\frac{M-1+2\epsilon}{2(M+\epsilon)}}
\end{eqnarray}
where $\xi=\min(\frac{1}{\xi_1},\frac{1}{\xi_2})$.

To cancel interference at destinations, similar to the linear scheme we choose the following alignment
\begin{eqnarray}
G_{11}v_{R_1,i+1}&=&-G_{12}v_{R_2,i}
\label{eq:alignd1rational}\\
-G_{21}v_{R_1,i}&=&G_{22}v_{R_2,i}
~~~~~~~~~\forall i\in\{1,\ldots,M-1\}\label{eq:alignd2rational}.
\end{eqnarray}
From \eqref{eq:alignd1rational} and \eqref{eq:alignd2rational}, it can be easily obtained that
\begin{eqnarray}
v_{R_1,i+1}&=&\left(G_{11}^{-1}G_{12}G_{22}^{-1}G_{21}\right)^iv_{R_1,1}\\
v_{R_2,i}&=& -\left(G_{22}^{-1}G_{21}G_{11}^{-1}G_{12}\right)^{i-1}G_{22}^{-1}G_{21}v_{R_1,1}
\end{eqnarray}
Again, once $v_{R_1,1}$ is determined, then all other scaling factors can be calculated using above equations. We choose $v_{R_1,1}=(G_{11}G_{22})^{M-1}$. Thus, we have
\begin{eqnarray}
v_{R_1,i+1}&=&(G_{12}G_{21})^{i}(G_{11}G_{22})^{M-i-1}\\
v_{R_2,i}&=&-G_{11}^{M-i}G_{12}^{i-1}G_{21}^iG_{22}^{M-1-i}~~~i\in\{1,\cdots, M-1\}
\end{eqnarray}

\subsubsection*{Destinations:}
After aligned interference cancellation, each destination can decode its desired signals. The received signal at $D_1$ is
\begin{eqnarray}
Y_1&=&G_{11}X_{R_1}+G_{12}X_{R_2}+N_1\notag\\
&=&BG_{11}\sum_{k_1=1}^{M}v_{R_1,k_1}\hat{x}_{R_1,k_1}+BG_{12}\sum_{k_2=1}^{M-1}v_{R_2,k_2}\hat{x}_{R_2,k_2}+N_1\notag\\
&\stackrel{(a)}{=}&BG_{11}v_{R_1,1}\hat{x}_{R_1,1}+B\sum_{i=1}^{M-1}G_{11}v_{R_1,i+1}\underbrace{(\hat{x}_{R_1,i+1}-\hat{x}_{R_2,i})}_{x_{D_1,i+1}}+N_1
\end{eqnarray}
where $(a)$ uses the alignment condition \eqref{eq:alignd1rational}. Let $x_{D_1,1}=\hat{x}_{R_1,1}$. Again, since $v_{R_1,1},\cdots, v_{R_1,M}$ are distinct monomial functions of channel coefficients, from \cite{Motahari_Gharan_Maddah_Khandani}, it can be shown that  $D_1$ can estimate $x_{D_1,k_1}$ with error probability $\text{Pr}\{\hat{x}_{D_1,k_1}\neq x_{D_1,k_1}\}$ going to zero as $P\rightarrow \infty$. After estimating $x_{D_1,k_1}$, $D_1$ estimates $x_{1,k_1}$ using the following estimator
\begin{eqnarray}
\hat{x}_{1,k_1}=\sum_{m=1}^{k_1}\hat{x}_{D_1,m} ~~~~~~~~\forall k_1\in\{1,\cdots, M\}
\end{eqnarray}
Then the message $W_{1,k_1}$ is decoded using a jointly typical decoder using a block of $\hat{x}_{1,k_1}$. Then $W_{1,k_1}$ can achieve a rate:
\begin{eqnarray}
I(x_{1,k_1};\hat{x}_{1,k_1})&=&H(x_{1,k_1})-H(x_{1,k_1}|\hat{x}_{1,k_1})\\
&\stackrel{(a)}{\ge}& \log|\mathcal{C}|-1-\text{Pr}(\hat{x}_{1,k_1}\neq x_{1,k_1})\log|\mathcal{C}|\\
&=&(1-\text{Pr}(\hat{x}_{1,k_1}\neq x_{1,k_1}))\log|\mathcal{C}|-1
\end{eqnarray}
where $(a)$ uses Fano's inequality.  Note that $\text{Pr}(\hat{x}_{1,k_1}\neq x_{1,k_1})$ will go to zero as  $P\rightarrow \infty$,  since  $\text{Pr}(\hat{x}_{R_1,k_1}\neq x_{R_1,k_1})$, $\text{Pr}(\hat{x}_{R_2,k_2}\neq x_{R_2,k_2})$ and $\text{Pr}\{\hat{x}_{D_1,k_1}\neq x_{D_1,k_1}\}$ go to zero as $P\rightarrow \infty$. And $\log|\mathcal{C}|=\frac{1-\epsilon}{2(M+\epsilon)}\log P + o(\log P)$. Therefore, as $P\rightarrow \infty$, $W_{1,k_1}$ achieves a rate equal to $\frac{1-\epsilon}{2(M+\epsilon)}\log P + o(\log P)$ and thus $\frac{1-\epsilon}{M+\epsilon}$ DoF.

Similarly, the received signal at $D_2$ is
\begin{eqnarray}
Y_2&=&G_{21}X_{R_1}+G_{22}X_{R_2}+N_2\notag\\
&=&BG_{21}\sum_{k_1=1}^{M}v_{R_1,k_1}\hat{x}_{R_1,k_1}+BG_{22}\sum_{k_2=1}^{M-1}v_{R_2,k_2}\hat{x}_{R_2,k_2}+N_2\notag\\
&\stackrel{(a)}{=}&BG_{21}v_{R_1,M}\hat{x}_{R_1,M}+B\sum_{i=1}^{M-1}G_{22}v_{R_2,i}(\hat{x}_{R_2,i}-\hat{x}_{R_1,i})+N_2
\end{eqnarray}
where $(a)$ uses the alignment condition \eqref{eq:alignd2rational}. Similar to $D_1$, $D_2$ can decode $W_{2,k_2}$, $k_2\in\{1,\cdots,M-1\}$, each carrying $\frac{1-\epsilon}{M+\epsilon}$ DoF. Therefore, the total number of DoF is $(2M-1)\frac{1-\epsilon}{M+\epsilon}$. Since $M$ can be made arbitrarily large and $\epsilon$ can be made arbitrarily small, we can achieve arbitrarily close to 2 DoF.

\subsection{Constant channel coefficients with linear scheme}
In last section, we show that 2 DoF can be achieved almost surely when the channels are held constant. The achievable scheme is based on lattice codes within the framework of rational dimensions. In this section, we show that with linear beamforming scheme and constant channels, at least $\frac{3}{2}$ DoF can be achieved for almost all channel coefficients. The key to this result is asymmetric complex signaling introduced in \cite{Cadambe_Jafar_Wang}. Let us denote the complex channel $F_{kj}=|F_{kj}|e^{i\phi_{kj}}$ and $G_{kj}=|G_{kj}|e^{i\theta_{kj}}$. We also use an alternative representation for \eqref{eq:firsthop} and \eqref{eq:secondhop} in terms of only real quantities as
\begin{eqnarray*}
\left[\begin{array}{c}\text{Re}\{Y_{R_k}\}\\ \text{Im}\{Y_{R_k}\}\end{array}\right]&=&\sum_{j=1}^2|F_{kj}|\underbrace{\left[\begin{array}{cc}\cos(\phi_{kj})& -\sin(\phi_{kj})\\ \sin(\phi_{kj})& \cos(\phi_{kj})\end{array}\right]}_{\mathbf{U}(\phi_{kj})}\left[\begin{array}{c}\text{Re}\{X_{j}\}\\ \text{Im}\{X_{j}\}\end{array}\right]+\left[\begin{array}{c}\text{Re}\{Z_{k}\}\\ \text{Im}\{Z_{k}\}\end{array}\right]\\
\left[\begin{array}{c}\text{Re}\{Y_{k}\}\\ \text{Im}\{Y_{k}\}\end{array}\right]&=&\sum_{j=1}^2|G_{kj}|\underbrace{\left[\begin{array}{cc}\cos(\theta_{kj})& -\sin(\theta_{kj})\\ \sin(\theta_{kj})& \cos(\theta_{kj})\end{array}\right]}_{\mathbf{U}(\theta_{kj})}\left[\begin{array}{c}\text{Re}\{X_{R_j}\}\\ \text{Im}\{X_{R_j}\}\end{array}\right]+\left[\begin{array}{c}\text{Re}\{N_{k}\}\\ \text{Im}\{N_{k}\}\end{array}\right]~~~~k\in\{1,2\}
\end{eqnarray*}
As a result, the original SISO complex channel becomes a real $2\times 2$ MIMO channel with scaled rotation channel matrices. The linear scheme proposed in Section \ref{sec:Timevarying} can be applied here. The only difference is that instead of diagonal channel matrices as in the SISO time-varying case, here we have scaled rotation channel matrices.  In order for the scheme work, we need to ensure linear independence of $\mathbf{v}_{1,1}$ and $\mathbf{v}_{1,2}$ in the first hop and linear independence of $\mathbf{v}_{R_1,1}$ and $\mathbf{v}_{R_1,2}$ in the second hop. From \eqref{eq:v1}, we have
\begin{eqnarray}
\mathbf{v}_{1,2}&=&|F_{11}|^{-1}|F_{12}||F_{22}|^{-1}|F_{21}|\mathbf{U}(\phi_{12}+\phi_{21}-\phi_{11}-\phi_{22})\mathbf{v}_{1,1}
\end{eqnarray}
Thus, in order for $\mathbf{v}_{1,2}$ and $\mathbf{v}_{1,1}$ to be linearly independent, $\mathbf{U}(\phi_{12}+\phi_{21}-\phi_{11}-\phi_{22})$ cannot be the identity matrix, which is ensured if the following condition is satisfied:
\begin{eqnarray}
\phi_{12}+\phi_{21}-\phi_{11}-\phi_{22} \neq 0 \mod (\pi)
\end{eqnarray}
Similarly, in the second hop, $\mathbf{v}_{R_1,1}$ and $\mathbf{v}_{R_1,2}$ are linearly independent if the following condition is satisfied:
\begin{eqnarray}
\theta_{12}+\theta_{21}-\theta_{11}-\theta_{22} \neq 0 \mod (\pi)
\end{eqnarray}
If both conditions are satisfied, $3$ real DoF can be achieved for the real MIMO channel and thus $\frac{3}{2}$ complex DoF are achieved for the original complex SISO channel. Thus we obtain the following theorem.
\begin{theorem}\label{thm:asy}
The $2 \times 2 \times 2$ IC with constant complex channel coefficients achieves $\frac{3}{2}$ DoF if all the following conditions are satisfied:
\begin{eqnarray}
\phi_{12}+\phi_{21}-\phi_{11}-\phi_{22} &\neq& 0 \mod (\pi)\\
\theta_{12}+\theta_{21}-\theta_{11}-\theta_{22} &\neq& 0 \mod (\pi)
\end{eqnarray}
\end{theorem}

\begin{corollary}\label{corollary:asy}
For the $2 \times 2 \times 2$ IC with constant complex channel coefficients, at least $\frac{3}{2}$ DoF can be achieved with linear scheme for almost all values of channel coefficients.
\end{corollary}
\begin{proof}
Since the exceptions of conditions in Theorem \ref{thm:asy} represent a subset of channel coefficients of measure 0, Theorem \ref{thm:asy} implies Corollary \ref{corollary:asy}.
\end{proof}

\section{Extensions}
In previous sections, we show that the min-cut outer bound of 2 DoF can be achieved for the $2\times 2\times 2$ interference network for generic channel coefficients. The result can be easily extended to more than two hops. Therefore, for the two sources and two destinations layered multihop interference network, regardless of the number of hops, the min-cut 2 DoF can be achieved almost surely. To see this, suppose we have $M$ hops. Let all relays after the second hop simply amplify and forward their signals with generic amplification factors. This reduces the network to an effective 2-hop setting. Hence, the DoF results apply. In fact, for more than 2 hops, aligned interference neutralization may not be needed to achieve 2 DoF. Simply amplify and forward received signals at relays may achieve exactly 2 DoF. To see this, let us consider again why exactly 2 DoF cannot be achieved with amplify and forward received signals at relays for 2 hops. From \eqref{eq:neu1} and \eqref{eq:neu2}, two equations need to be satisfied for interference neutralization at both destinations. Without loss of generality, $\alpha_1$ can be normalized to be one, leaving one variable $\alpha_2$. Since there are two equations and one variable, the problem does not admit a solution for generic channel coefficients. However, as the number of hops increases, more than one variable (amplification factor) are involved. Since the number of multivariate polynomial equations that needs to be satisfied is always 2, one at each destination which requires the effective coefficient of the interfering symbol equal to zero, the problem may have solutions.

Another extension of the result is to the multiple-input multiple-output (MIMO) channel where each node is equipped with $M$ antennas. The linear scheme proposed for SISO time-varying case with $M$ symbol extensions can be directly carried over to the MIMO case. The only difference is that for MIMO channel, the channel matrices are no long in a diagonal form. As a result, to ensure $\mathbf{v}_{1,1},\cdots,\mathbf{v}_{1,M}$ and  $\mathbf{v}_{R_1,1},\cdots,\mathbf{v}_{R_1,M}$ are linearly independent, it can be shown that  $\mathbf{F}_{11}^{-1}\mathbf{F}_{12}\mathbf{F}_{22}^{-1}\mathbf{F}_{21}$ and $\mathbf{G}_{11}^{-1}\mathbf{G}_{12}\mathbf{G}_{22}^{-1}\mathbf{G}_{21}$ should have $M$ distinct eigenvalues. Thus, with these two conditions satisfied, $2M-1$ DoF can be achieved for the $2\times 2\times 2$ MIMO interference networks.

\section{Conclusion}
We explore the DoF for the $2\times 2\times 2$ interference channel. We show that the min-cut outer bound value of 2 DoF can be achieved for almost all channel coefficients regardless of whether the channel is time-varying or constant. The key to this result is a new idea called aligned interference neutralization which combines the ideas of interference alignment and interference neutralization. For the time-varying case our aligned interference neutralization scheme is based on linear beamforming schemes over symbol extensions. The same scheme is translated into the rational dimensions framework for the case of constant channel coefficients. This is particularly interesting because the rational dimensions framework has been used previously for interference alignment, but not for interference neutralization. One limitation of the rational dimensions framework is that the DoF guarantees are provided for almost all values of channel coefficients but cannot be made explicitly for any \emph{given} realization of channel coefficients. Linear beamforming schemes on the other hand allow explicit DoF guarantees for given channel realizations. To provide explicit DoF guarantees with constant channels, we apply the asymmetric complex signaling scheme with aligned interference neutralization to the $2\times 2\times 2$ IC and find explicit sufficient conditions on the channel coefficients to guarantee that the channel has at least $\frac{3}{2}$ DoF. Interestingly, these conditions only depend on the phases of the channels and are obtained through asymmetric complex signaling with linear beamforming scheme\cite{Cadambe_Jafar_Wang}.

The result for $2\times 2\times 2$ IC is also extended to more than 2-hop IC with 2 sources and 2 destinations, for which the min-cut outer bound 2 DoF can still be achieved for almost all channels. Thus, regardless of the number of hops, 2 sources and 2 destinations multihop IC has 2 DoF almost surely.  Extensions to the setting with more than 2 sources and 2 destinations are challenging. This is an interesting question to be pursued in future work.

\bibliography{Thesis}

% Generated by IEEEtran.bst, version: 1.12 (2007/01/11)
\begin{thebibliography}{10}
\providecommand{\url}[1]{#1}
\csname url@samestyle\endcsname
\providecommand{\newblock}{\relax}
\providecommand{\bibinfo}[2]{#2}
\providecommand{\BIBentrySTDinterwordspacing}{\spaceskip=0pt\relax}
\providecommand{\BIBentryALTinterwordstretchfactor}{4}
\providecommand{\BIBentryALTinterwordspacing}{\spaceskip=\fontdimen2\font plus
\BIBentryALTinterwordstretchfactor\fontdimen3\font minus
  \fontdimen4\font\relax}
\providecommand{\BIBforeignlanguage}[2]{{%
\expandafter\ifx\csname l@#1\endcsname\relax
\typeout{** WARNING: IEEEtran.bst: No hyphenation pattern has been}%
\typeout{** loaded for the language `#1'. Using the pattern for}%
\typeout{** the default language instead.}%
\else
\language=\csname l@#1\endcsname
\fi
#2}}
\providecommand{\BIBdecl}{\relax}
\BIBdecl

\bibitem{Avestimehr_Diggavi_Tse}
A.~S. Avestimehr, S.~Diggavi, and D.~Tse, ``A deterministic approach to
  wireless relay networks,'' oct 2007, arXiv:cs.IT/0710.3777.

\bibitem{Jafar_Shamai}
S.~Jafar and S.~Shamai, ``Degrees of freedom region for the {MIMO} {X}
  channel,'' \emph{IEEE Trans. on Information Theory}, vol.~54, no.~1, pp.
  151--170, Jan. 2008.

\bibitem{Cadambe_Jafar_int}
V.~Cadambe and S.~Jafar, ``Interference alignment and the degrees of freedom of
  the {K} user interference channel,'' \emph{IEEE Trans. on Information
  Theory}, vol.~54, no.~8, pp. 3425--3441, Aug. 2008.

\bibitem{Cadambe_Jafar_X}
------, ``Interference alignment and the degrees of freedom of wireless x
  networks,'' \emph{IEEE Trans. on Information Theory}, no.~9, pp. 3893--3908,
  Sep 2009.

\bibitem{Motahari_Gharan_Maddah_Khandani}
A.~S. Motahari, S.~O. Gharan, M.~A. Maddah-Ali, and A.~K. Khandani, ``Forming
  pseudo-mimo by embedding infinite rational dimensions along a single real
  line: Removing barriers in achieving the dofs of single antenna systems,''
  \emph{CoRR}, vol. abs/0908.2282, 2009.

\bibitem{Parker_Bliss_Tarokh}
P.~Parker, D.~Bliss, and V.~Tarokh, ``On the degrees-of-freedom of the {MIMO}
  interference channel,'' in \emph{42nd Annual Conference on Information
  Sciences and Systems (CISS)}, March 2008.

\bibitem{Jafar_Vishwanath_GDOF}
``Generalized degrees of freedom of the symmetric {Gaussian} k user
  interference channel,'' in \emph{arXiv:cs.IT/0804.4489}, 2008.

\bibitem{Bandemer_ElGamal}
B.~Bandemer, A.~El~Gamal, and G.~Vazquez-Vilar, ``{On the sum capacity of a
  class of cyclically symmetric deterministic interference channels},'' in
  \emph{Information Theory, 2009. ISIT 2009. IEEE International Symposium
  on}.\hskip 1em plus 0.5em minus 0.4em\relax IEEE, 2009, pp. 2622--2626.

\bibitem{Huang_Cadambe_Jafar}
C.~Huang, V.~Cadambe, and S.~Jafar, ``On the capacity and generalized degrees
  of freedom of the x channel,'' October 2008, arxiv:0810.4741.

\bibitem{Gou_Jafar_O1}
T.~Gou and S.~A. Jafar, ``{Capacity of a class of symmetric SIMO Gaussian
  interference channels within O (1)},'' in \emph{IEEE International Symposium
  on Information Theory, ISIT}.\hskip 1em plus 0.5em minus 0.4em\relax IEEE,
  2009, pp. 1924--1928.

\bibitem{Gou_Jafar_Wang}
T.~Gou, S.~Jafar, and C.~Wang, ``Degrees of freedom of finite state compound
  wireless networks,'' \emph{To Appear in the IEEE Transactions on Information
  Theory. Full paper available at arXiv:0909.4203}, 2009.

\bibitem{Etkin_Tse_Wang}
R.~Etkin, D.~Tse, and H.~Wang, ``Gaussian interference channel capacity to
  within one bit,'' \emph{submitted to IEEE Transactions on Information
  Theory}, Feb. 2007.

\bibitem{Bresler_Parekh_Tse}
G.~Bresler, A.~Parekh, and D.~Tse, ``Approximate capacity of the many-to-one
  interference channel,'' \emph{39th Annual Allerton Conference on
  Communication, Control and Computing}, Sep. 2007.

\bibitem{Suh_Tse_FB}
C.~Suh and D.~Tse, ``Feedback capacity of the gaussian interference channel to
  within 1.7075 bits: the symmetric case,'' \emph{CoRR}, vol. abs/0901.3580,
  2009.

\bibitem{Tse_Yates}
D.~N.~C. Tse and R.~D. Yates, ``Fading broadcast channels with state
  information at the receivers,'' \emph{CoRR}, vol. abs/0904.3165, 2009.

\bibitem{MK_int}
A.~Motahari and A.~Khandani, ``Capacity bounds for the {Gaussian} interference
  channel,'' in \emph{arXiv:cs/0801.1306 [cs.IT]}, 2008.

\bibitem{Shang_Kramer_Chen}
X.~Shang, G.~Kramer, and B.~Chen, ``A new outer bound and the
  noisy-interference sum-rate capacity for gaussian interference channels,''
  \emph{submitted to IEEE Transactions on Information Theory. Preprint
  available on Arxiv arXiv:0712.1987}, Dec. 2007.

\bibitem{Sreekanth_Veeravalli}
V.~Annapureddy and V.~Veeravalli, ``Gaussian interference networks: Sum
  capacity in the low interference regime and new outer bounds on the capacity
  region,'' in \emph{Submitted to IEEE Transactions on Information Theory.
  arxiv:0802.3495}, Feb 2008.

\bibitem{Sreekanth_Veeravalli_MIMO}
V.~S. Annapureddy and V.~V. Veeravalli, ``Sum capacity of mimo interference
  channels in the low interference regime,'' \emph{CoRR}, vol. abs/0909.2074,
  2009.

\bibitem{Sridharan_Jafarian_Vishwanath_Jafar}
S.~Sridharan, A.~Jafarian, S.~Vishwanath, and S.~Jafar, ``Capacity of symmetric
  {K}-user gaussian very strong interference channels,'' in \emph{Proceedings
  of {IEEE GLOBECOM}}, Dec 2008.

\bibitem{Jose_Vishwanath}
J.~Jose and S.~Vishwanath, ``Sum capacity of k user gaussian unit rank
  interference channels,'' \emph{CoRR}, vol. abs/1004.2104, 2010.

\bibitem{Cadambe_Jafar_MACZBC}
V.~Cadambe and S.~Jafar, ``{Sum-capacity and the unique separability of the
  parallel Gaussian MAC-Z-BC network},'' in \emph{Information Theory
  Proceedings (ISIT), 2010 IEEE International Symposium on}.\hskip 1em plus
  0.5em minus 0.4em\relax IEEE, 2010, pp. 2318--2322.

\bibitem{Jafar_ergodic}
S.~A. Jafar, ``The ergodic capacity of interference networks,'' \emph{CoRR},
  vol. abs/0902.0838, 2009.

\bibitem{Borade_Zheng_Gallager}
S.~Borade, L.~Zheng, and R.~Gallager, ``Maximizing degrees of freedom in
  wireless networks,'' in \emph{Proceedings of 40th Annual Allerton Conference
  on Communication, Control and Computing}, October 2003, pp. 561--570.

\bibitem{Tannious_Nosratinia}
R.~Tannious and A.~Nosratinia, ``{The interference channel with MIMO relay:
  Degrees of freedom},'' in \emph{IEEE International Symposium on Information
  Theory, ISIT 2008.}\hskip 1em plus 0.5em minus 0.4em\relax IEEE, 2008, pp.
  1908--1912.

\bibitem{Chen_Cheng}
S.~Chen and R.~Cheng, ``Achieve the degrees of freedom of {K-User} {MIMO}
  interference channel with a {MIMO} relay,'' in \emph{IEEE GLOBECOM}, Dec.
  2010.

\bibitem{Boelcskei_Nabar_Oyman_Paulraj}
H.~Boelcskei, R.~Nabar, O.~Oyman, and A.~Paulraj, ``Capacity scaling laws in
  mimo relay networks,'' \emph{Trans. on Wireless Communications}, vol.~5,
  no.~6, pp. 1433--1444, June 2006.

\bibitem{Morgenshtern_Boelcskei}
\BIBentryALTinterwordspacing
V.~Morgenshtern and H.~B\"olcskei, ``Crystallization in large wireless
  networks,'' \emph{IEEE Transactions on Information Theory}, vol.~53, no.~10,
  pp. 3319--3349, Oct. 2007. [Online]. Available:
  \url{http://www.nari.ee.ethz.ch/commth/pubs/p/transit06}
\BIBentrySTDinterwordspacing

\bibitem{Esli_Wittneben}
\BIBentryALTinterwordspacing
C.~Esli and A.~Wittneben, ``A hierarchical {AF} protocol for distributed
  orthogonalization in multiuser relay networks,'' \emph{IEEE Transactions on
  Vehicular Technology}, 2010. [Online]. Available:
  \url{http://www.nari.ee.ethz.ch/wireless/pubs/p/esli_tvt_2010}
\BIBentrySTDinterwordspacing

\bibitem{Rankov_Wittneben}
\BIBentryALTinterwordspacing
B.~Rankov and A.~Wittneben, ``Spectral efficient protocols for half-duplex
  fading relay channels,'' \emph{IEEE Journal on Selected Areas in
  Communications}, Feb. 2007. [Online]. Available:
  \url{http://www.nari.ee.ethz.ch/wireless/pubs/p/jsac2006}
\BIBentrySTDinterwordspacing

\bibitem{Wittneben_mag}
\BIBentryALTinterwordspacing
S.~Berger, T.~Unger, M.~Kuhn, A.~Klein, and A.~Wittneben, ``Recent advances in
  amplify-and-forward two-hop relaying,'' \emph{Communications Magazine}, July
  2009. [Online]. Available:
  \url{http://www.nari.ee.ethz.ch/wireless/pubs/p/CommMag08}
\BIBentrySTDinterwordspacing

\bibitem{Simeone_mesh}
O.~Simeone, O.~Somekh, Y.~Bar-Ness, H.~V. Poor, and S.~Shamai, ``Capacity of
  linear two-hop mesh networks with rate splitting, decode-and-forward relaying
  and cooperation,'' \emph{CoRR}, vol. abs/0710.2553, 2007.

\bibitem{Thejasvi_Calderbank_Cochran}
P.~Thejaswi, A.~Bennatan, J.~Zhang, R.~Calderbank, and D.~Cochran,
  ``{Rate-achievability strategies for two-hop interference flows},'' in
  \emph{46th Annual Allerton Conference on Communication, Control, and
  Computing}.\hskip 1em plus 0.5em minus 0.4em\relax IEEE, 2008, pp.
  1432--1439.

\bibitem{Cao_Chen}
Y.~Cao and B.~Chen, ``Capacity bounds for two-hop interference networks,''
  \emph{CoRR}, vol. abs/0910.1532, 2009.

\bibitem{Nosratinia_Madsen}
A.~Host-Madsen and A.~Nosratinia, ``The multiplexing gain of wireless
  networks,'' in \emph{Proc. of ISIT}, 2005.

\bibitem{Birk_Kol}
Y.~Birk and T.~Kol, ``{Informed-source coding-on-demand (ISCOD) over broadcast
  channels},'' in \emph{INFOCOM'98. Seventeenth Annual Joint Conference of the
  IEEE Computer and Communications Societies. Proceedings. IEEE}, vol.~3.\hskip
  1em plus 0.5em minus 0.4em\relax IEEE, 2002, pp. 1257--1264.

\bibitem{MMK}
M.~{Maddah-Ali}, A.~Motahari, and A.~Khandani, ``Communication over {MIMO X}
  channels: Interference alignment, decomposition, and performance analysis,''
  in \emph{IEEE Trans. on Information Theory}, August 2008, pp. 3457--3470.

\bibitem{Weingarten_Shamai_Kramer}
H.~Weingarten, S.~Shamai, and G.~Kramer, ``On the compound {MIMO} broadcast
  channel,'' in \emph{Proceedings of Annual Information Theory and Applications
  Workshop UCSD}, Jan 2007.

\bibitem{Nazer_Gastpar_Jafar_Vishwanath}
B.~Nazer, M.~Gastpar, S.~Jafar, and S.~Vishwanath, ``Ergodic interference
  alignment,'' in \emph{ISIT}, 2009.

\bibitem{Cadambe_Jafar_Wang}
V.~Cadambe, S.~Jafar, and C.~Wang, ``{Interference Alignment With Asymmetric
  Complex SignalingÑSettling the H{\o}st-Madsen--Nosratinia Conjecture},''
  \emph{IEEE Transactions on Information Theory}, vol.~56, no.~9, pp.
  4552--4565, 2010.

\bibitem{Etkin_Ordentlich}
R.~Etkin and E.~Ordentlich, ``On the degrees-of-freedom of the {K}-user
  {Gaussian} interference channel,'' in \emph{arXiv:cs.IT/0901.1695}, Jan 2009.

\bibitem{Cadambe_Jafar_Shamai}
V.~Cadambe, S.~Jafar, and S.~Shamai, ``Interference alignment on the
  deterministic channel and application to {Gaussian} networks,'' \emph{IEEE
  Trans. on Information Theory}, vol.~55, no.~1, pp. 269--274, Jan 2009.

\bibitem{Jafar_corr}
S.~Jafar, ``Exploiting channel correlations -- simple interference alignment
  schemes with no {CSIT},'' \emph{Preprint available on Arxiv arXiv:0910.0555},
  Oct. 2009.

\bibitem{Maddah_Tse}
\BIBentryALTinterwordspacing
M.~A. Maddah-Ali and D.~Tse, ``On the degrees of freedom of {MISO} broadcast
  channels with delayed feedback,'' EECS Department, University of California,
  Berkeley, Tech. Rep. UCB/EECS-2010-122, Sep 2010. [Online]. Available:
  \url{http://www.eecs.berkeley.edu/Pubs/TechRpts/2010/EECS-2010-122.html}
\BIBentrySTDinterwordspacing

\bibitem{Cadambe_Jafar_Maleki}
V.~R. Cadambe, S.~A. Jafar, and H.~Maleki, ``Distributed data storage with
  minimum storage regenerating codes - exact and functional repair are
  asymptotically equally efficient,'' \emph{CoRR}, vol. abs/1004.4299, 2010.

\bibitem{Motahari_Gharan_Khandani_real}
A.~Motahari, S.~{Oveis Gharan}, and A.~Khandani, ``Real interference alignment
  with real numbers,'' Aug 2009, arXiv:0908.1208.

\bibitem{Mohajer_Diggavi_Fragouli_Tse}
S.~Mohajer, S.~Diggavi, C.~Fragouli, and D.~Tse, ``{Transmission techniques for
  relay-interference networks},'' in \emph{46th Annual Allerton Conference on
  Communication, Control, and Computing}.\hskip 1em plus 0.5em minus
  0.4em\relax IEEE, 2008, pp. 467--474.

\bibitem{Mohajer_Tse_Diggavi}
S.~Mohajer, D.~Tse, and S.~Diggavi, ``{Approximate capacity of a class of
  Gaussian relay-interference networks},'' in \emph{IEEE International
  Symposium on Information Theory, ISIT}.\hskip 1em plus 0.5em minus
  0.4em\relax IEEE, 2009, pp. 31--35.

\bibitem{Gomadam_Jafar_corr}
K.~Gomadam and S.~Jafar, ``{The effect of noise correlation in
  amplify-and-forward relay networks},'' \emph{IEEE Transactions on Information
  Theory}, vol.~55, no.~2, pp. 731--745, 2009.

\bibitem{Jeon_Chung}
S.~{-W. Jeon} and S.~{-Y. Chung}, ``Capacity of a class of multi-source relay
  networks,'' in \emph{Proceedings of UCSD Workshop on Information Theory and
  Applications}, Feb 2009.

\bibitem{Jeon_Chung_Jafar}
S.~{-W. Jeon}, S.~{-Y. Chung}, and S.~Jafar, ``Degrees of freedom of
  multi-source wireless relay networks,'' in \emph{Proceedings of Allerton
  conference}, Sep 2009.

\bibitem{SPCOM_Plenary}
S.~A. Jafar, ``{On asymptotic interference alignment: Plenary talk},'' in
  \emph{2010 International Conference on Signal Processing and Communications
  (SPCOM)}.\hskip 1em plus 0.5em minus 0.4em\relax IEEE, 2010, pp. 1--5.

\bibitem{Maddah_Compound}
M.~A. Maddah-Ali, ``The degrees of freedom of the compound {MIMO} broadcast
  channels with finite states,'' \emph{CoRR}, vol. abs/0909.5006, 2009.

\end{thebibliography}

\end{document}